\documentclass[11pt,fleqn]{article}
\usepackage[T1]{fontenc}
\usepackage[margin=1in]{geometry}
\usepackage{amsfonts,amsmath,amssymb,amsthm,mathtools,graphicx,setspace}
\usepackage[dvipsnames]{xcolor}
\usepackage[colorlinks=true,linkcolor=blue,citecolor=Mahogany,urlcolor=blue]{hyperref}
\usepackage[normalem]{ulem}
\usepackage[longnamesfirst]{natbib}
\usepackage{titling}
\usepackage{bbm}
\usepackage{enumitem}
\usepackage{etoolbox}
\usepackage[noamsmath]{kpfonts}
\usepackage{inconsolata}
\usepackage{natbib}

\definecolor{accentcolor}{RGB}{26,148,49}
\hypersetup{pdfborderstyle={},allcolors=accentcolor,colorlinks}
\numberwithin{equation}{section}
\allowdisplaybreaks
\setlist[enumerate]{topsep=3pt,itemsep=3pt,parsep=0pt}
\setlist[itemize]{topsep=3pt,itemsep=3pt,parsep=0pt}

\newcommand{\zerodisplayskips}{%
  \setlength{\abovedisplayskip}{8pt}%
  \setlength{\belowdisplayskip}{7pt}%
  \setlength{\abovedisplayshortskip}{8pt}%
  \setlength{\belowdisplayshortskip}{7pt}}
\appto{\normalsize}{\zerodisplayskips}
\appto{\small}{\zerodisplayskips}
\appto{\footnotesize}{\zerodisplayskips}

\newtheorem{theorem}{Theorem}
\newtheorem{proposition}{Proposition}
\newtheorem{corollary}{Corollary}
\newtheorem{lemma}{Lemma}
\newtheorem{definition}{Definition}
\newtheorem{assumption}{Assumption}

\newcommand{\R}{\mathbb R}

\newcommand{\Nzero}{\mathbb N_0}
\newcommand{\E}{\mathbf E}
\newcommand{\A}{\mathcal A}
\newcommand{\Ccal}{\mathcal C}
\newcommand{\Dcal}{\mathcal D}

\newcommand{\Kcal}{\mathcal K}
\newcommand{\Lcal}{\mathcal L}

\newcommand{\e}{\mathbf e}
\newcommand{\1}{\mathbf 1}

\DeclareMathOperator*{\argmax}{arg\,max}

\pretitle{\begin{flushleft}\LARGE\bfseries}
	\posttitle{\par\end{flushleft}\vspace{-0.5em}}

\preauthor{\begin{flushleft}\large}
	\postauthor{\par\end{flushleft}\vspace{-1em}}

\predate{\begin{flushleft}\normalsize}
	\postdate{\par\end{flushleft}}

\title{Schedules and Prioritization: A Behavioral Foundation for Multi-Armed Bandits and Stopping Problems}

\author{%
	Jaden Yang Chen\thanks{UNC Chapel Hill. Email: \href{mailto:yangch@unc.edu}{yangch@unc.edu}.}
	\quad
	Can Urgun\thanks{UNC Chapel Hill. Email: \href{mailto:curgun@unc.edu}{curgun@unc.edu}.}
}

\date{\today}

\begin{document}
\maketitle

\begin{abstract}
Bandit models typically begin with arms, states, rewards, and transition rules. This paper instead begins with preferences over stopped local contingent schedules: possible unfoldings of a responsibility, project, experiment, or opportunity in its own local time. Behavioral axioms on single schedules characterize a generalized stopping representation with current utility, local discounting, and a broad continuation aggregator. A common-tail compensation axiom then allows calendar time to be priced across schedules. Imposing a tight elapsed-calendar constraint generates a rested generalized bandit and yields index optimality: the index is the shadow price of advancing a local clock. Expected-utility, learning, robust, rank-dependent, Choquet, and Pandora models arise as special cases.

\vspace{2mm}
\end{abstract}

\newpage
\pagenumbering{arabic}

\section{Introduction}

A person has many strands of responsibility in her life. There may be a work strand: answer a referee report, revise a draft, wait for feedback, decide whether the next step is a small edit or a larger rewrite. There may be a child-related strand: take a child to school, see how the school day went, check homework, learn from a report card that mathematics needs attention, decide whether to help directly or arrange outside help. There may be a household or caregiving strand: make dinner, schedule a repair, follow up with a family member, learn that a small task has generated a larger one. There may be a relationship strand: plan something thoughtful, respond to how it is received, and decide what kind of further attention is called for. These examples are deliberately ordinary. The point is that a strand of responsibility is not a single dated payoff. It is a contingent sequence of consequences.

The larger problem is to fit many such strands into one calendar. A day has only so many slots. Work, household, family, and relationship responsibilities cannot all be advanced at once. But neither are they consumed and discarded like independent goods. If a work file is put aside while a family matter is handled, the work file remains there. If a household repair is postponed while a deadline is met, the repair does not vanish. Each strand has its own internal time: the first time it is advanced, the second time it is advanced, the next time something is learned, and the point at which that strand is completed or abandoned. Calendar time is the scarce resource that determines which of these local clocks is allowed to move.

The reason to begin with these mundane strands is the same reason consumer theory begins with consumption bundles. Demand is a workhorse object of economics, but a demand function is not usually taken as primitive. One starts with preferences over the right domain, obtains a utility representation from behavioral restrictions, and then derives demand by imposing a budget constraint. The domain and the constraint do different jobs. Bundles are the objects the consumer can rank; the budget is the scarcity constraint that turns those rankings into demand.

This paper develops an analogous foundation for bandit problems. Bandits are used throughout economics to study experimentation, learning, search, attention, project selection, medical trials, and dynamic resource allocation. Yet the standard formulation usually begins one step downstream. The analyst gives the decision maker a collection of arms, specifies states, rewards, transition rules, and an objective, and then asks how the arms should be operated. That is often the right applied model. It is not, however, a behavioral foundation for why such a bandit problem is the right representation of the decision maker's situation.

Our starting point is that, in many economic settings, what the decision maker sees is not a Markov machine called an arm. She sees a responsibility, project, experiment, or opportunity that may unfold in different ways as attention is devoted to it. The primitive alternatives are therefore possible unfoldings of one strand of responsibility. We call these objects local contingent schedules. They are local because they are written in the internal time of the strand, before the strand is placed into a common calendar with other strands. They are contingent because later consequences may depend on what has happened so far. They are stopped because, along every realized branch, the strand eventually ends.

The first part of the paper studies one strand in isolation. Before asking whether work, family, household, or some other responsibility should receive the next unit of calendar time, we ask how the decision maker evaluates possible continuations of a single responsibility. We impose two groups of behavioral restrictions. The deterministic-stream axioms apply when the strand has no remaining branching. They identify the ordinary intertemporal part of the evaluation: a common scale for current consequences and a common discount factor over local time. The recursive one-step axioms apply when the strand has a current consequence and a contingent folder of future continuations. They require that equivalent continuations can be substituted, that branchwise improvements are respected, and that current consequences and continuation prospects can be traded off in a stable way.

The first main theorem says that these restrictions are equivalent to a recursive stopping representation. This is a strong equivalence, not merely a sufficient condition for a convenient formula. Recursive representations are fairly common in decision theory. The novelty is that, on our stopped-schedule domain, the usual behavioral restrictions generate stopping problems. The representation says that advancing a strand delivers some utility today, future consequences are discounted in local time, and the remaining uncertainty is evaluated by an aggregator over the possible continuation values. That aggregator is deliberately broad: in special cases it becomes expected utility, multiplier or variational preferences, rectangular max--min, rank-dependent or Choquet preferences, and other evaluations of future branches.

This is the first sense in which the paper is foundational. Stopping theory usually begins with a process, a payoff, and a set of admissible stopping rules. The analyst writes down what the state is, how it evolves, and what the decision maker receives upon stopping. Here, the stopped tree of consequences is the primitive alternative. The theorem identifies when preferences over those trees can be treated as if the decision maker were solving a stopping problem. Thus the stopping problem is not imposed at the outset. It is represented from preferences over stopped schedules.

This one-strand representation is the analogue of utility in consumer theory. It tells us how the decision maker evaluates each primitive object. But utility alone is not demand, and a stopping representation alone is not yet a bandit. The second step is to add the relevant scarcity constraint. For consumers, the constraint is a budget. For responsibility strands, the constraint is the calendar: only one local clock can move at a time.

Before calendar time can be priced, however, one more behavioral issue must be addressed. A priority index is a scalar object. It compares the claims that different strands have on the next unit of calendar time. For this comparison to be meaningful, scalar charges and compensations must pass through the evaluation of continuation folders without being distorted. This is automatic under ordinary expected utility, but it is not automatic when continuations are evaluated nonlinearly. Ambiguity, robustness, or rank dependence can make a common shift in all future branches fail to behave like the same common shift after aggregation.

We capture the needed property with a common-tail compensation axiom. Suppose two ways of handling the beginning of a strand are exactly balanced. Now append the same downstream continuation behind both of them. Since the appended continuation is common, it should not alter the original compensation. In the responsibility-strand language, if two ways of opening a file are equally good before the same future file is attached, attaching that same future file should not make one opening suddenly better than the other. The axiom is behavioral, but its economic role is to make scalar shadow prices legitimate inside local schedule problems.

This is the bridge from stopping to bandits. Once scalar calendar-time charges can be moved through the continuation evaluation, represented preferences over local schedules can be combined with a tight calendar constraint. The elapsed-calendar constraint says that local clocks cannot in aggregate run faster than calendar time. Each calendar period advances one strand, while the others remain fixed. This is the scarce-resource constraint that turns many local responsibility schedules into one prioritization problem.

The central result is the Bandit Representation and Index Theorem. Its first part says that the elapsed-calendar problem has the form of a generalized rested multi-armed bandit. The state of the decision maker is the collection of current local states, one for each open strand. Advancing one strand changes only that strand's local state; all other strands are frozen. The arms of the bandit are therefore not primitive Markov machines. They are represented responsibility strands placed under a common calendar constraint.

The second part says that the optimal policy is an index policy. The index of a strand at a given local state is the critical shadow price of calendar time at which that strand is just worth putting on hold. This is a priority price, not an additional behavioral primitive. A strand with a high index is one for which the next local step is expensive to postpone, because the decision maker would still be willing to advance it even when calendar time is priced highly. The optimal rule is to advance a strand with the highest current index. Thus index optimality is derived from the domain, the axioms, and the calendar constraint, rather than added as an independent rule.

The paper then shows how familiar models arise from additional behavioral or domain restrictions. Under ordinary independence for continuation schedules, the continuation aggregator is linear. The local schedule is then represented by a subjective transition law, and the index reduces to the classical Gittins--Jones index. A learning model is obtained by treating beliefs as part of the local state. Advancing a strand produces both a current consequence and a signal about what kind of strand it is. Resting the strand leaves its belief fixed, because no local signal arrives while attention is elsewhere.

With variational or multiplier-style preferences, the decision maker is not willing to evaluate a continuation folder under a single forecast. Instead, alternative probabilistic descriptions are considered, with penalties attached to less plausible descriptions. Max--min preferences are the stark limiting case. There, the relevant adverse description is not merely a single bad one-step transition chosen today. It is closer to a series of unfortunate events: an unfavorable rectangular description of how the strand may continue to evolve is selected against the contemplated stopping plan. The resulting index is a robust priority price for advancing the strand under dynamic ambiguity.

With rank-dependent or Choquet preferences, the decision maker evaluates continuation branches by their rank rather than only by additive probabilities. This is natural in the schedule domain. A work strand may command attention because a very promising response would create a valuable continuation; a school matter may command attention because a bad report would require immediate intervention; a household repair may loom large because one branch turns it into a serious problem. In a static lottery this looks like rank dependence over prizes. In a continuation problem it looks closer to fear of missing out: the cost of putting a strand on hold includes the possibility of failing to reach a salient future branch.

Pandora's problem arises in a different way. It is not primarily a different attitude toward continuation uncertainty. It is a restriction on the schedule domain. A box is a strand whose uncertainty is resolved in one local step. Once the box is opened, the revealed value is absorbing. Under this domain restriction, the general index becomes a reservation value. Thus Pandora's rule is another instance of the same schedule logic: a special kind of responsibility strand generates a special kind of priority price.

The unifying message is that bandit models can be built from the same ingredients that underlie other economic choice problems: a primitive domain, behavioral restrictions that yield a representation, and a scarcity constraint. A transition law, an ambiguity set, a capacity, or a reservation value is not the starting point. Each appears after adding further restrictions to preferences or to the schedule domain. Contingent schedules therefore play for bandits the role that consumption bundles play for demand: they are the primitive domain on which behavioral restrictions are imposed before scarcity is introduced.

The paper is organized as follows. The next subsection briefly reviews the related literature. Section~\ref{sec:model} defines local contingent schedules and the one-step composition operation. Section~\ref{sec:axioms} states the deterministic-stream and recursive one-step axioms and proves the recursive stopping representation. Section~\ref{sec:cash} introduces common-tail compensation and its scalar-charge implication. Section~\ref{sec:multi} establishes the Bandit Representation and Index Theorem. Section~\ref{sec:extensions} develops the subjective expected-utility, learning, variational, max--min, rank-dependent, Choquet, and Pandora special cases. Section~7 concludes. All proofs are relegated to the appendix.

\subsection{Literature}

Bandit models are one of the standard ways economists represent dynamic allocation under uncertainty. They appear in experimentation, learning, search, attention, project choice, medical trials, and other settings in which choosing an alternative today both produces a payoff and changes what can be learned or done tomorrow. Early economic applications include pricing under demand uncertainty and strategic experimentation \citep{Rothschild1974,BoltonHarris1999,KellerRadyCripps2005}; \citet{BergemannValimaki2008} surveys the experimentation literature. The classical Gittins--Jones theory gives the main solution result once the bandit problem has been specified: a discounted allocation problem with projects that evolve only when selected can be reduced to single-project stopping problems, and an optimal policy selects a project with maximal index \citep{GittinsJones1974,Gittins1979,Whittle1980,BerryFristedt1985,KatehakisVeinott1987,Weber1992,GittinsGlazebrookWeber2011}. These papers solve bandits once the relevant objects are already in place. The decision maker is handed arms governed by known stochastic processes, and in learning applications she is handed an information structure that determines which observations arrive when an alternative is tried and how beliefs update. Our question is prior to that specification. In the same way that revealed-preference theory disciplines demand by starting from preferences over a domain of objects rather than from an arbitrary demand curve \citep{Samuelson1938,Afriat1967,Varian1982}, we ask what primitive domain and what refutable behavioral restrictions make a dynamic choice problem representable as a bandit. In our model, the primitive objects are local stopped schedules; the bandit appears only after those schedules have a stopping representation and are placed under one elapsed-calendar constraint.

The same issue appears already at the level of stopping. Optimal stopping problems are ubiquitous in economics and finance: agents decide when to invest, exercise, liquidate, abandon, inspect, or accept an alternative. Classical stopping theory begins with a stochastic process and asks when it should be stopped, with Snell-envelope and martingale methods providing the standard tools \citep{ChowRobbinsSiegmund1971,Shiryaev1978,PeskirShiryaev2006}; real-options models use the same structure to study investment timing under uncertainty \citep{DixitPindyck1994}. Search models in the Pandora tradition combine ordering and stopping: \citet{Weitzman1979} gives the classical reservation-value rule, and subsequent work studies how far Pandora-style rules extend beyond the original environment \citep{OlszewskiWeber2015,Doval2018}. A particularly relevant recent contribution is \citet{AusterChe2025}, who revisit Pandora's problem under robust search and connect robust search behavior to choice overload. These papers begin with processes, boxes, inspection costs, or reservation equations. Our single-strand theorem begins instead with a stopped object in the decision maker's domain and identifies when preferences over such objects have a stopping representation. Pandora's problem then appears as a special schedule domain, in which opening a box resolves uncertainty in one local step and the revealed value is absorbing.

A further literature takes the bandit object as given and changes the way uncertainty is evaluated. \citet{Anderson2012} shows how ambiguity aversion can reduce the value of experimentation. \citet{Li2019} develops a multiple-priors bandit model and obtains Gittins--Jones-style results. \citet{CohenTreetanthiploet2022} prove a Gittins theorem under nonlinear expectations. Robust and risk-sensitive bandit models modify the objective, reward process, or uncertainty set, and identify conditions under which index-like rules survive or fail \citep{KimLim2016,CaroGupta2022,DenardoParkRothblum2007,MalekipirbazariCavus2024}. This literature is useful for us because it makes clear that the location of uncertainty is substantive. If ambiguity is placed on the full multi-project environment, the adverse model may choose a joint description of all projects. Projects that are technologically separate in the classical bandit can then become linked through the worst-case joint law. In our setting, this is not the natural interpretation of the underlying economic problem. A work obligation, a household repair, a medical follow-up, and a family responsibility may all compete for the same calendar, but uncertainty about one strand need not be tied to uncertainty about the others. We therefore allow rich ambiguity, robustness, and rank dependence within each local continuation problem, while keeping the cross-strand interaction exactly where the schedule interpretation puts it: in the scarce calendar. The product structure is not an independence assumption imposed on an already-formed multi-arm Markov process; it is the consequence of representing distinct local schedules separately before they are woven into one calendar.

Finally, the axiomatic part of the paper belongs to the tradition of dynamic decision theory and recursive preferences. Koopmans's work on stationary ordinal utility and impatience, and the Koopmans--Diamond--Williamson extension, provide classic foundations for stationary intertemporal utility and geometric discounting \citep{Koopmans1960,KoopmansDiamondWilliamson1964}. Kreps and Porteus study temporal lotteries and the timing of uncertainty resolution, while Epstein and Zin develop recursive, not-necessarily expected-utility preferences over intertemporal consumption lotteries \citep{KrepsPorteus1978,EpsteinZin1989}. Chew and Epstein axiomatize recursive utility under uncertainty and integrate atemporal non-expected-utility models into a dynamic framework \citep{ChewEpstein1991}. Related work shows how learning, information, updating, and ambiguity primitives can themselves be disciplined by choice: \citet{DillenbergerLlerasSadowskiTakeoka2014} study subjective learning when information arrival is unobserved by the analyst, \citet{DillenbergerKrishnaSadowski2023} model subjective constraints on information acquisition and absorption, \citet{Denti2022} gives testable restrictions for posterior-separable information costs, and \citet{Cheng2022} axiomatizes relative maximum-likelihood updating for ambiguous beliefs. For ambiguity and nonadditive beliefs, the relevant antecedents also include max--min expected utility, Choquet expected utility, variational preferences, recursive multiple priors, smooth ambiguity, multiplier preferences, dynamic mixture aversion, adaptive preferences, dynamic \(\alpha\)-MEU, and dual-self representations of ambiguity \citep{AnscombeAumann1963,GilboaSchmeidler1989,Schmeidler1989,MaccheroniMarinacciRustichini2006,MaccheroniMarinacciRustichini2006Dynamic,EpsteinSchneider2003,KlibanoffMarinacciMukerji2005,HansenSargent2001,HansenSargent2008,DentiPomatto2022,Sarver2018,SadowskiSarver2024,FrickIijimaLeYaouanq2022,ChandrasekherFrickIijimaLeYaouanq2022}. Our first representation theorem is deliberately in this recursive-preference tradition: a one-step schedule is evaluated by combining a current consequence with an aggregator of continuation values. The novelty is not recursion per se. The difference is the domain and the use of the representation. Standard recursive preferences are usually defined over global temporal consumption lotteries, streams, or acts; our primitives are stopped local schedules, or responsibility strands, whose clocks advance only when selected. This local-time domain is what makes a bandit interpretation possible. The common-tail axiom then supplies the central additional step: it implies cash additivity of continuation values, which allows elapsed calendar time to be priced by a scalar charge. When many local schedules share one calendar clock, these scalar charges become indices, and advancing a maximal-index strand is optimal.

\section{A Model of Contingent Schedules}\label{sec:model}

We start by introducing a local contingent schedule, which broadly corresponds to an ordered flow of consequences generated by a strand of responsibility. We use one notational convention throughout. Local states are treated as consequence-bearing states. Formally, a schedule is a map from local histories into one-period consequences, so the state and the assigned consequence could be kept as separate objects. For simplicity we suppress that extra projection. Thus, a local state records the payoff-relevant situation of the strand, including the current one-period consequence generated at that local situation. When later formulas write \(u(x)\) or \(u(x_i)\), this means the utility of the current consequence attached to that local state. This convention is only a relabeling.\footnote{D8 below fixes the common affine scale on which current consequences are compared; the objects that may remain state- and arm-dependent are the continuation operators.}

For a single strand, let \(S=\{1,\ldots,n\}\) be a finite local state space. For an initial state \(x\in S\), define
\[
H_t(x):=\{(x_0,\ldots,x_t)\in S^{t+1}:x_0=x\},
\qquad
H(x):=\bigcup_{t\ge0}H_t(x).
\]
The set \(H_t(x)\) is the set of possible histories after \(t\) local steps of the strand. For an infinite local history \(\omega=(x_0,x_1,\ldots)\) with \(x_0=x\), write \(\omega^t=(x_0,\ldots,x_t)\). Let \(\Kcal\) be the space of one-period consequences, and let \(\bot\) denote the terminal symbol indicating that the schedule has ended.

\begin{definition}[Local contingent schedule]\label{def:local-schedule}
A local contingent schedule at \(x\) is a map
\[
s:H(x)\to \Kcal\cup\{\bot\}
\]
satisfying two restrictions.
\begin{enumerate}[label=(\roman*)]
\item \emph{Absorbing nullity.} If \(s(h_t)=\bot\), then \(s(h_s)=\bot\) for every extension \(h_s\supset h_t\).
\item \emph{Finite termination.} For every infinite history \(\omega\) starting at \(x\), there exists a finite date \(T_s(\omega)\) such that \(s(\omega^t)\in\Kcal\) for \(t<T_s(\omega)\) and \(s(\omega^t)=\bot\) for \(t\ge T_s(\omega)\).
\end{enumerate}
\end{definition}

Let \(\A_x\) denote the set of schedules at \(x\). A schedule specifies what
consequence is received after each possible local history, until that branch
ends. Absorbing nullity says that an ended branch stays ended. Finite termination
says that, along each branch, only finitely many nonterminal consequences are
generated. This is the domain on which the primitive preference relation
\(\succeq_x\) is imposed.

The finite-termination convention is behavioral rather than ontological. The
axioms are imposed on stopped schedules because these are the finite objects over
which choices can in principle be observed, compared, and refuted. This is the
same discipline behind revealed-preference restrictions such as WARP: the
primitive restrictions are placed on finite choice objects, even if the
representation later assigns values on a larger mathematical domain. Thus the
restriction to stopped schedules should not be read as saying that the decision
maker cannot contemplate infinite streams or that the analyst cannot represent
their values. Once the representation is obtained, finite schedules determine a
natural extension to infinite deterministic streams and infinite-horizon
recursive problems. That extension is generated by discounted limits and by
concatenating continuation schedules at finite tails.


The point of writing the domain this way is that the uncertainty of a strand is not imposed on the decision maker as a fully specified Markov process or signal structure. What is primitive is the comparison between possible stopped trees of consequences. A work strand, for example, may branch depending on feedback; a child-related strand may branch depending on what is learned from school; a household strand may branch depending on whether a repair turns out to be minor or serious. The representation result asks when these comparisons can be summarized by familiar ingredients: a utility scale for current consequences, discounting across local time, and a continuation rule that captures both the possible next local situations and the decision maker's attitude toward the uncertainty in those continuations. In the expected-utility special case, that continuation rule is a subjective transition kernel; in the more general cases, it may be a nonlinear aggregator reflecting ambiguity, robustness, or risk-sensitive evaluation. Thus the usual stochastic process, transition kernel, or learning rule appears only after additional structure is imposed on schedule preferences, rather than being part of the primitive description \citep{Rothschild1974,BoltonHarris1999,KellerRadyCripps2005,BergemannValimaki2008}.

The one-period consequence space can be specialized according to the source of uncertainty. In the baseline formulation below, \(\Kcal=\Lcal:=\Delta(Z)\), the space of objective lotteries over prizes, and \(0\in\Lcal\) is a null lottery. If ambiguity also concerns one-period rewards, \(\Kcal\) can instead be taken to be a space of horse lotteries or subjective one-period acts. The recursive representation is stated on the lottery-flow specialization.

A deterministic finite stream is a tuple \(c=(\ell_0,\ldots,\ell_{T-1})\in\Lcal^T\). Let
\[
\Ccal:=\bigcup_{T\in\Nzero}\Lcal^T.
\]
The deterministic schedule \(d_x(c)\) assigns \(\ell_t\) to every history-date node at local date \(t<T\) and assigns \(\bot\) at all later local dates. Thus a deterministic schedule is the same sequence of consequences on every branch of the local tree. Let
\[
\Dcal_x:=\{d_x(c):c\in\Ccal\},
\]
and write \(d_x(())\) for immediate stopping. We normalize \(d_x(())\sim_x d_x((0))\) for every \(x\).

For each \(x\in S\), preferences over schedules starting at \(x\) are denoted by \(\succeq_x\), with symmetric and asymmetric parts \(\sim_x\) and \(\succ_x\). The subscript records the local state at which the strand is evaluated. The same formal continuation may be evaluated differently when it is reached from different local states, because the continuation is attached to a different local situation.

\subsection{One-step composition}

The recursive structure is built from a simple operation: isolate the current consequence of a strand from the continuation that follows the next local state. Given a current lottery \(\ell\in\Lcal\) and continuation schedules \(a^y\in\A_y\) for all next states \(y\in S\), define
\[
\Gamma_x(\ell,(a^y)_{y\in S})\in\A_x
\]
by assigning \(\ell\) at the initial history \((x)\) and, after a first transition to \(y\), following \(a^y\) from that point onward:
\[
\Gamma_x(\ell,(a^y)_y)(x,x_1,\ldots,x_t)=a^{x_1}(x_1,\ldots,x_t),
\qquad t\ge1.
\]
A one-step schedule therefore has two parts: a current consequence and a
contingent collection of future schedules, one for each possible next local
state. Repeated applications of \(\Gamma_x\) build finite local trees and attach
continuation schedules at their terminal histories. This is the local object from
which the recursive representation is built.

The first result below asks what restrictions on preferences over such schedules make this recursive description possible. The answer will be that a single strand can be evaluated as a stopping problem. Only after that one-strand problem is understood do we return to the larger calendaring problem: several strands of responsibility, each with its own local time, must be woven into a single elapsed calendar.

\section{Recursive Stopping Preferences}\label{sec:axioms}

We now put a behavioral/axiomatic structure on preferences over a single strand. At this point there is still no calendar problem. The decision maker is not yet choosing between work, household, family, or other strands. She is evaluating possible ways in which one strand may unfold: different stopped trees of consequences, all written in that strand's own local time.


The first group of axioms concerns deterministic strands. These are schedules with no branching: the same finite sequence of consequences is obtained along every history. They are useful because they pin down the ordinary intertemporal tradeoffs that are already present before uncertainty over the continuation tree enters, in the spirit of Koopmans-style foundations for discounted utility \citep{Koopmans1960,KoopmansDiamondWilliamson1964}. The second group of axioms concerns one-step contingent schedules. These ask when the decision maker's comparison between ``one consequence now'' and ``a contingent continuation from tomorrow onward'' can be summarized recursively, as in the recursive-preference tradition \citep{KrepsPorteus1978,EpsteinZin1989,ChewEpstein1991}.

\begin{assumption}[Deterministic-stream axioms D]\label{ass:D}
For each $x\in S$, the restriction of $\succeq_x$ to $\Dcal_x$ satisfies the following conditions.
\begin{enumerate}[label=\textup{D\arabic*.},leftmargin=1.25cm]
\item \emph{Weak order.} The relation is complete and transitive.
\item \emph{Mixture independence on common horizons.} For fixed $T$, streams $c,d,e\in\Lcal^T$, and $\alpha\in(0,1)$,
\[
d_x(c)\succeq_x d_x(d)
\quad\Longrightarrow\quad
 d_x\bigl(\alpha c+(1-\alpha)e\bigr)\succeq_x d_x\bigl(\alpha d+(1-\alpha)e\bigr),
\]
where mixtures are coordinatewise.
\item \emph{Same-length tail independence.} For fixed $T$, streams $c,d\in\Lcal^T$, and tails $h,k\in\Ccal$,
\[
d_x(c\oplus h)\succeq_x d_x(d\oplus h)
\quad\Longleftrightarrow\quad
 d_x(c\oplus k)\succeq_x d_x(d\oplus k).
\]
\item \emph{Stationarity.} For all $c,d\in\Ccal$,
\[
d_x(0\oplus c)\succeq_x d_x(0\oplus d)
\quad\Longleftrightarrow\quad
d_x(c)\succeq_x d_x(d).
\]
\item \emph{Zero-tail irrelevance.} For every $c\in\Ccal$, $d_x(c\oplus(0))\sim_x d_x(c)$.
\item \emph{Impatience.} There exist lotteries $\ell,m\in\Lcal$ such that $d_x((\ell))\succ_x d_x((m))$ and $d_x((\ell,m))\succ_x d_x((m,\ell))$.
\item \emph{Continuity.} For each $T$, the restriction to $\Lcal^T$ is continuous in the product topology.
\item \emph{State-independence on deterministic streams.} For all $x,y\in S$ and $c,d\in\Ccal$,
\[
d_x(c)\succeq_x d_x(d)
\quad\Longleftrightarrow\quad
d_y(c)\succeq_y d_y(d).
\]
\end{enumerate}
\end{assumption}

These assumptions say that when a strand has no remaining uncertainty, the decision maker evaluates its finite flow of consequences in the usual discounted-utility way. D1 gives coherent rankings. D2 is the source of affine utility on fixed-horizon lottery streams. D3 says that if two front pieces are compared before a common tail, the particular common tail should not change the ranking of those front pieces. D4 says that adding a common blank local period at the beginning does not change a comparison, which is what delivers stationarity across local dates. D5 says that appending a blank period after a schedule has ended is not a new consequence. D6 rules out indifference to timing and pins the discount factor below one. D7 is continuity. D8 is the cross-state normalization: deterministic consequence streams are evaluated on the same one-period scale no matter which local state the strand is written from.

\begin{proposition}[Discounted utility on deterministic finite streams]\label{prop:detstreams}
Under Assumption~\ref{ass:D}, there exist an affine vNM index $u:\Lcal\to\R$ and a discount factor $\beta\in(0,1)$ such that deterministic schedules are represented by
\begin{equation}\label{eq:det-rep}
U_x(d_x(\ell_0,\ldots,\ell_{T-1}))
=
\sum_{t=0}^{T-1}\beta^t u(\ell_t)
\qquad \text{for every }x\in S.
\end{equation}
The pair $(u,\beta)$ is unique up to positive affine rescaling of $u$.
\end{proposition}

The proposition pins down the one-period utility scale and the local discount
factor on branchless schedules. We next turn to one-step schedules, where the
current consequence is followed by a contingent folder of continuation schedules. Let
\[
Q:=\prod_{y\in S}(\A_y/{\sim_y})
\]
be the product of continuation equivalence classes. For $[a]=([a^y])_{y\in S}\in Q$, define the induced one-step preference $\succeq_x^1$ on $\Lcal\times Q$ by
\[
(\ell,[a])\succeq_x^1(m,[b])
\quad\Longleftrightarrow\quad
\Gamma_x(\ell,(a^y)_{y\in S})\succeq_x \Gamma_x(m,(b^y)_{y\in S}).
\]
The quotient notation emphasizes the idea that once the first local step is taken, a continuation schedule matters through how it is evaluated from the next local state onward.

\begin{assumption}[Recursive one-step axioms R]\label{ass:R}
For each $x\in S$, one-step comparisons satisfy the following conditions.
\begin{enumerate}[label=\textup{R\arabic*.},leftmargin=1.25cm]
\item \emph{Conditional substitution of equivalents.} If $a^y\sim_y\tilde a^y$ for every $y$, then
\[
\Gamma_x(\ell,(a^y)_y)\sim_x\Gamma_x(\ell,(\tilde a^y)_y)
\qquad \text{for every }\ell\in\Lcal.
\]
\item \emph{Branchwise monotonicity.} If $a^y\succeq_y b^y$ for every $y$, then
\[
\Gamma_x(\ell,(a^y)_y)\succeq_x\Gamma_x(\ell,(b^y)_y)
\qquad \text{for every }\ell\in\Lcal.
\]
\item \emph{Weak order and continuity on one-step schedules.} The induced relation $\succeq_x^1$ is complete, transitive, and continuous on $\Lcal\times Q$.
\item \emph{Coordinate independence.} Current lotteries and continuation profiles are separately independent:
\begin{align*}
(\ell,q)\succeq_x^1(m,q)&\Longleftrightarrow(\ell,q')\succeq_x^1(m,q')
&&\forall \ell,m\in\Lcal,\ \forall q,q'\in Q,\\
(\ell,q)\succeq_x^1(\ell,q')&\Longleftrightarrow(m,q)\succeq_x^1(m,q')
&&\forall \ell,m\in\Lcal,\ \forall q,q'\in Q.
\end{align*}
\item \emph{Double cancellation.} If
\[
(\ell_1,q_1)\succeq_x^1(\ell_2,q_2)
\quad\text{and}\quad
(\ell_2,q_3)\succeq_x^1(\ell_3,q_1),
\]
then
\[
(\ell_1,q_3)\succeq_x^1(\ell_3,q_2).
\]
\item \emph{Solvability.} Intermediate one-step values can be matched either by changing the current lottery while holding a continuation fixed or by changing the continuation while holding a current lottery fixed: if
\[
(\ell_1,q)\succeq_x^1(\ell_2,q')\succeq_x^1(\ell_3,q),
\]
then $(\tilde\ell,q)\sim_x^1(\ell_2,q')$ for some $\tilde\ell\in\Lcal$; and if
\[
(\ell,q_1)\succeq_x^1(m,q_2)\succeq_x^1(\ell,q_3),
\]
then $(\ell,\tilde q)\sim_x^1(m,q_2)$ for some $\tilde q\in Q$.
\end{enumerate}
\end{assumption}

The R axioms say that the decision maker's treatment of one local step is internally consistent. R1 says that continuation schedules that are already equivalent from the relevant next state can be substituted branch by branch. R2 says that improving every possible continuation branch improves the one-step schedule. R3 gives coherent and continuous rankings on the one-step domain. R4 separates the ranking of current consequences from the ranking of continuation profiles: changing the common continuation should not reverse the ranking of two current consequences, and changing the common current consequence should not reverse the ranking of two continuation profiles. R5 is the corresponding consistency of compensations between current and future parts of the strand. R6 is the richness condition that lets those compensations be found inside the primitive domain.

The result is not that the strand was secretly an expected-utility Markov process all along. The result is that, under D and R, the decision maker behaves as if a single strand can be evaluated recursively: current utility is evaluated on the scale identified above, and the future tree is summarized by a continuation operator. That continuation operator may later reduce to a subjective transition law, a set of priors, a variational penalty, or another ambiguity-sensitive object; at this level it is simply the way the decision maker evaluates the possible continuations of the strand.

\begin{theorem}[Recursive stopping representation]\label{thm:recursive}
Suppose the interval-richness required for solvability holds on the relevant attainable sets. Then preferences satisfy Assumptions~\ref{ass:D} and~\ref{ass:R} if and only if there exist numerical representations $U_x:\A_x\to\R$, a common affine one-period utility index $u:\Lcal\to\R$, a common discount factor $\beta\in(0,1)$, and continuous increasing continuation operators
\[
\rho_x:\prod_{y\in S}U_y(\A_y)\to\R
\]
such that deterministic schedules satisfy \eqref{eq:det-rep} and every one-step schedule satisfies
\begin{equation}\label{eq:recursive-rep}
U_x\Bigl(\Gamma_x(\ell,(a^y)_{y\in S})\Bigr)
=u(\ell)+\beta\rho_x\bigl(U_1(a^1),\ldots,U_n(a^n)\bigr),
\end{equation}
where the one-period index is the same affine scale as in \eqref{eq:det-rep}.
\end{theorem}

The theorem is the first main reduction. Primitive preferences over trees of consequences are equivalent to a recursive evaluation of a single schedule. The current consequence enters through the common utility index \(u\). The continuation enters only through the vector of continuation values and the state-dependent aggregator \(\rho_x\). Thus the theorem separates what is common across different strands of responsibility and states, namely the utility scale for current consequences, from what may vary with the local schedule corresponding to a strand of responsibility, namely the evaluation of continuation uncertainty.

The theorem can be read as a schedule-accounting result. The proof first fixes the
ruler for a single strand when nothing is uncertain. For a branchless local file, the D axioms say
that the value cannot depend on irrelevant descriptions of the same stopped schedule, only on
which consequences occur at which local dates. They therefore calibrate every deterministic local
schedule on the common discounted scale in Proposition \ref{prop:detstreams}. This is the calibration that must also be used once the strand becomes uncertain.  

Now consider a one-step uncertain schedule. In the schedule language, it is a card with two
fields: the consequence produced when the file is advanced now, and the contingent folder of
schedules left for the next local state. Rather than adding substantive assumptions about probabilities, the R axioms impose a bookkeeping discipline on this card. They ensure that equivalent continuation folders can be substituted, that branchwise improvements are respected, and
that current consequences and continuation folders have stable compensating tradeoffs. In other
words, the current field and the continuation field can be separated without changing the underlying ordering of schedules. That separation is exactly what is needed for an additive representation of a one-step card:
\[
A_x(\ell)+B_x\bigl((a^y)_y\bigr).
\]
The deterministic representation then identifies the first coordinate. If all continuations are deterministic one-date schedules, the representation must agree with the already-calibrated value of a two-date deterministic schedule; therefore \(A_x\) is the same affine consequence index \(u\), and the continuation coordinate is on the \(\beta\)-discounted scale. What remains unrestricted is how local state \(x\) summarizes the vector of possible continuation values. Calling that summary \(\rho_x\) gives
\[
U_x\bigl(\Gamma_x(\ell,(a^y)_y)\bigr)
= u(\ell)+\beta\rho_x\bigl((U_y(a^y))_{y\in S}\bigr).
\]
Monotonicity and continuity of \(\rho_x\) are inherited from the corresponding schedule axioms. Thus the axioms turn a primitive preference over local schedules into a recursive local accounting
rule: score what the strand delivers now, then add a discounted assessment of the file that remains.


The representation also explains why the word ``stopping'' is the right one, but
it should be read locally. Classical stopping theory typically begins with a
stochastic process and an admissible class of stopping times, and then asks when
the process should be stopped
\citep{ChowRobbinsSiegmund1971,Shiryaev1978,PeskirShiryaev2006}. Here the order
is reversed. The primitive object is a local schedule, and the axioms D and R are
\emph{equivalent} to a generalized stopping representation of preferences over
such schedules. Stopping does not mean that the person has literally abandoned
the real-world responsibility, or that no other strand can be attended to. It
means that, in the isolated one-strand problem, the local schedule is no longer
advanced and no further consequences from that schedule are counted. Later, in
the multi-schedule problem, not advancing a strand at a calendar date is different
from stopping it: the calendar may simply be used on another strand.

For a finite history \(h=(x_0,\ldots,x_t)\), write $	\bar x(h):=x_t$
for its current local state. If
\(g=(x_t,y_1,\ldots,y_m)\in H_m(x_t)\) is a continuation history starting from
\(\bar x(h)\), write
$
	h\circ g:=(x_0,\ldots,x_t,y_1,\ldots,y_m)
$
for the history obtained by grafting \(g\) after \(h\). In particular, for
\(y\in S\), write $hy:=(x_0,\ldots,x_t,y)$.
Let $H(x):=\bigcup_{t\ge0}H_t(x)$
be the set of finite local histories starting from \(x\).

\begin{corollary}[Induced stopping representation]\label{cor:induced-stopping}
	Assume the hypotheses of Theorem~\ref{thm:recursive}, and use the normalization
	\(u(0)=0\). Under the consequence-bearing-state convention, the completed
	one-strand stopping value starting from \(x\) is represented by the minimal
	bounded nonnegative function \(V_x:H(x)\to\R_+\) satisfying
	\begin{equation}\label{eq:history-snell}
		V_x(h)
		=
		\max\left\{
			0,\,
			u(\bar x(h))
			+
			\beta
			\rho_{\bar x(h)}
			\left(
				\bigl(V_x(hy)\bigr)_{y\in S}
			\right)
		\right\}.
	\end{equation}
	Equivalently,
	\begin{equation}\label{eq:stopping-sup}
		V_x((x))
		=
		\sup_{\tau\le\infty}\mathcal U_x(\tau),
	\end{equation}
	where \(\tau\) ranges over local stopping rules and
	\(\mathcal U_x(\tau)\) is the completed recursive value of the schedule that
	advances the local history before \(\tau\) and assigns \(\bot\) thereafter.
\end{corollary}

Thus Theorem~\ref{thm:recursive} turns the local evaluation of a strand into a
generalized Snell envelope. At each finite local history, the decision maker
compares ending the schedule, with normalized value zero, to advancing the
schedule one more local step, receiving the consequence carried by the current
local state, and recursively evaluating the continuation folder indexed by the
next local state. The statement uses no exogenous probability space: the
continuation folder is evaluated by the represented operator
\(\rho_{\bar x(h)}\).

\section{Common Tails and Scalar Charges}\label{sec:cash}

The preceding section assigns a value to the schedule associated with a single strand of responsibility. This is already useful locally: it says how the decision maker evaluates one stopped tree of consequences when that schedule is considered on its own. The calendaring problem, however, will require something more. Several schedules will compete for the same unit of elapsed calendar time, and the value of moving one schedule forward must be comparable to the value of moving another schedule forward.

D8 gives a common scale for current consequences. It does not by itself say that scalar changes in continuation value remain scalar changes after the continuation tree is evaluated. This distinction matters because the representation in Theorem~\ref{thm:recursive} permits nonlinear continuation operators. A continuation tree may be evaluated through expected utility, max--min preferences, variational preferences, smooth ambiguity, or some other uncertainty-sensitive rule. In such environments, a number added to every possible continuation branch need not automatically remain that same number after the continuation operator is applied.

For the index argument, the needed invariance is exactly this scalar portability. If one unit of calendar time is eventually priced by a scalar charge, then that scalar must enter a one-schedule stopping problem as a scalar. It cannot be reshaped by the continuation operator or depend on the rest of the calendar. The next axiom gives the primitive content of this requirement.

\begin{assumption}[Common-tail compensation invariance, R7]\label{ass:R7}
	For every $x\in S$, lotteries $\ell,m,\bar\ell,\bar m\in\Lcal$, and continuation profiles
	\[
	r_y=(r_y^z)_{z\in S}\in\prod_{z\in S}\A_z,
	\qquad y\in S,
	\]
	if
	\begin{equation}\label{eq:R7-premise}
		d_x((\ell,\bar m))\sim_x d_x((m,\bar\ell)),
	\end{equation}
	then
	\begin{equation}\label{eq:R7-conclusion}
		\Gamma_x\bigl(\ell,(\Gamma_y(\bar m,r_y))_{y\in S}\bigr)
		\sim_x
		\Gamma_x\bigl(m,(\Gamma_y(\bar\ell,r_y))_{y\in S}\bigr).
	\end{equation}
\end{assumption}

The axiom is easiest to read as a compensation test. In the two-date comparison \eqref{eq:R7-premise}, the decision maker is exactly indifferent between two ways of distributing consequences across today and tomorrow. One schedule gives $\ell$ today and $\bar m$ tomorrow; the other gives $m$ today and $\bar\ell$ tomorrow. Thus the move from $\bar m$ to $\bar\ell$ is precisely the tomorrow-compensation for the move from $\ell$ to $m$ today.

R7 says that this compensation remains exact when the same contingent continuation is attached behind tomorrow. The common continuation may be uncertain. It may branch in many ways. It may be evaluated nonlinearly. But it is the same continuation in both schedules, and so it should not change the exchange rate between current utility and a common one-period-later compensation. In the responsibility-strand language, if two ways of handling the first two local steps of a schedule are exactly balanced, then appending the same downstream continuation of the strand should not disturb that balance.

\begin{proposition}[Common tails and cash additivity]\label{prop:cash}
	Assume the recursive representation in Theorem~\ref{thm:recursive}. Fix $x\in S$. R7 is equivalent to cash additivity of $\rho_x$ on all behaviorally attainable common scalar changes. In particular, whenever \eqref{eq:R7-premise} holds and
	\[
	a^y:=\Gamma_y(\bar m,r_y),
	\qquad
	b^y:=\Gamma_y(\bar\ell,r_y),
	\qquad
	v_y:=U_y(a^y),
	\qquad
	c:=u(\bar\ell)-u(\bar m),
	\]
	one has
	\begin{equation}\label{eq:cash-attainable}
		\rho_x(v+c\1)=\rho_x(v)+c.
	\end{equation}
	Consequently, if the attainable continuation domain is closed under common scalar translations, then R7 is equivalent to full cash additivity on that domain:
	\begin{equation}\label{eq:cash-full}
		\rho_x(v+c\1)=\rho_x(v)+c,
	\end{equation}
	whenever both $v$ and $v+c\1$ lie in the domain of $\rho_x$.
\end{proposition}

R7 is the axiom that makes schedule values behave like scalar amounts when a common tail is attached. First ignore branching. The premise \eqref{eq:R7-premise} says that two two-step schedules are exactly balanced: doing \(\ell\) now and \(\bar m\) next is just as good as doing \(m\) now and \(\bar\ell\) next. Since deterministic schedules have already been put on the common discounted ruler, this balance means
\[
u(\ell)+\beta u(\bar m)=u(m)+\beta u(\bar\ell).
\]
If \(c:=u(\bar\ell)-u(\bar m)\), then replacing \(\bar m\) by \(\bar\ell\) at the next local step is a \(c\)-unit improvement on the continuation-value scale, and the current consequence \(\ell\) is exactly the discounted compensation for using \(m\) instead.

Now put uncertainty back into the schedule by attaching the same downstream file behind that second local step. In every possible next local state \(y\), compare the continuation schedule that begins with \(\bar m\) and then follows \(r_y\) with the one that begins with \(\bar\ell\) and then follows the same \(r_y\). The only difference is the first consequence of that continuation file; the rest of the file is identical. Therefore the whole future folder has been shifted up by \(c\) in every branch, without changing its downstream structure.

R7 says that this common downstream structure cannot disturb the original compensation. Substituting the recursive representation into the conclusion of R7 gives indifference exactly when
\[
u(m)+\beta c+\beta\rho_x(v)
=
u(m)+\beta\rho_x(v+c\1),
\]
or equivalently
\[
\rho_x(v+c\1)=\rho_x(v)+c.
\]
Thus the common-tail axiom is precisely cash additivity on the scalar shifts that can be generated by primitive schedules. The converse is the same schedule story read backward: if a common \(c\)-shift passes through \(\rho_x\) unchanged, then the original two-step compensation remains exact after the same continuation file is appended. If the continuation domain contains all common scalar translations, the attainable version becomes full cash additivity.

\medskip

This proposition is the bridge from one-schedule stopping values to calendar-time charges. Theorem~\ref{thm:recursive} and Corollary~\ref{cor:induced-stopping} say that a schedule associated with one strand has a scalar stopping value. R7 says that scalar shifts of such values pass through the recursive continuation operator without distortion. That is the property needed when many local schedules are later placed into one calendar and elapsed calendar time is priced by a scalar charge.

In the classical index literature, especially Weber's charge formulation and Whittle's related subsidy formulation, the charge is not a new preference primitive; it is a scalar price attached to the local decision to continue or to leave a schedule inactive \citep{GittinsJones1974,Whittle1980,Weber1992,GittinsGlazebrookWeber2011}. Cash additivity is what lets that scalar remain a scalar inside each one-schedule stopping problem. Without it, local stopping representations may still exist, but a one-dimensional index need not summarize how different schedules should compete for calendar time.

\section{Bandit Representation and the Index}\label{sec:multi}\label{sec:index}

We now return to the larger calendaring problem. The previous sections studied the schedule associated with a single strand of responsibility. They asked how one local schedule is evaluated, when it is worth advancing in its own local time, and when scalar changes in its continuation value remain scalar. The real allocation problem begins when several such schedules are open at once. Work, household, family, and other strands may each have their own local state and their own local clock, but they must be fit into one elapsed calendar.

The passage from one schedule to many schedules adds a feasibility constraint, not a new preference axiom. Each local schedule is still evaluated by the one-schedule representation above. What changes is that local time is now scarce across schedules. A unit of calendar time can advance one local schedule; it cannot advance all of them at once. Thus the role played by a budget in a consumer problem is played here by an elapsed-calendar-time constraint: it describes which collections of local schedule steps can be afforded by a single calendar.

Fix \(N\) local schedules, indexed by \(i=1,\ldots,N\). In the bandit terminology below, each local schedule is an arm. Schedule \(i\) has finite local state space \(S_i\). The product state space is
\[
\mathbf S:=\prod_{i=1}^N S_i,
\qquad
x=(x_1,\ldots,x_N).
\]
For \(x\in\mathbf S\) and \(y_i\in S_i\), write
\[
(x_{-i},y_i):=(x_1,\ldots,x_{i-1},y_i,x_{i+1},\ldots,x_N).
\]
The product construction is rested: if schedule \(i\) is advanced, then only \(x_i\) changes. All other local states remain fixed.

The calendar constraint is written in terms of local times. For each \(n\in\Nzero\), set
\[
G^n
:=
\left\{
a=(a_1,\ldots,a_N)\in\Nzero^N:
\sum_{i=1}^N a_i=n
\right\},
\qquad
G:=\bigcup_{n\ge0}G^n.
\]
For \(a\in G^n\), let
\[
D_a:=\{a+\e_i:i=1,\ldots,N\}\subseteq G^{n+1}.
\]

\begin{definition}[Increasing path and elapsed calendar constraint]\label{def:increasing-path}
	An increasing path is a sequence \(\gamma=(\gamma_n)_{n\ge0}\) with values in \(G\) such that
	\[
	\gamma_0=0,
	\qquad
	\gamma_{n+1}\in D_{\gamma_n}
	\quad\text{for every }n\ge0.
	\]
	Writing \(\gamma_n=(T_1(n),\ldots,T_N(n))\), an increasing path satisfies the \textbf{elapsed calendar time constraint} if
	\begin{equation}\label{eq:clock-identities}
		T_i(0)=0,
		\qquad
		T_i(n+1)-T_i(n)\in\{0,1\},
		\qquad
		\sum_{i=1}^N\bigl(T_i(n+1)-T_i(n)\bigr)=1.
	\end{equation}
\end{definition}

Here \(T_i(n)\) is the amount of local time assigned to schedule \(i\) by calendar date \(n\). Equation \eqref{eq:clock-identities} says that each calendar period advances exactly one local schedule. Equivalently,
\[
\sum_{i=1}^N T_i(n)=n
\qquad\text{for every calendar date }n.
\]
This is the calendar-budget identity. The local schedules may have different local lengths and may end at different histories, but their aggregate use of local time cannot exceed the calendar time available to the decision maker.

\begin{definition}[Contingent multi-schedule]\label{def:multi-schedule}
	A contingent multi-schedule based at \(x=(x_1,\ldots,x_N)\) consists of one local contingent schedule for each schedule \(i\), written in schedule \(i\)'s own local time, and an increasing path \(\gamma\) satisfying the elapsed calendar constraint \eqref{eq:clock-identities}. At each calendar date, the unique schedule whose local clock advances takes one step in its local schedule; all other local schedules remain at their current local states.
\end{definition}


This is the feasible set of the multi-schedule problem. The primitive alternatives remain local schedules. The increasing path is the feasible allocation of elapsed calendar time across those schedules. In this sense, the analogy with consumer choice is exact at the level of structure: preferences over primitive objects are one input, and a feasibility constraint is the other. The theorem below says that, once the primitive objects are local schedules and the feasibility constraint is elapsed calendar time, the resulting problem is a generalized rested bandit. The connection with the classical Gittins--Jones theorem is therefore through representation rather than through primitives: the object to which the usual single-project stopping logic applies is generated from schedule preferences and the calendar constraint, rather than assumed at the outset \citep{GittinsJones1974,Gittins1979,Weber1992}.

\begin{theorem}[Bandit representation and index]\label{thm:bandit-index}\label{thm:product-bandit}\label{thm:index-optimality}
	For each schedule \(i\), suppose the local schedule preferences satisfy D, R, and R7, and suppose the attainable continuation domain is closed under common scalar translations. Write the resulting schedule-\(i\) recursive representation on a common utility scale and common discount factor \(\beta\in(0,1)\) as follows: advancing schedule \(i\) in local state \(x_i\) gives current utility \(u(x_i)\), and the continuation vector over next local states is evaluated by an increasing cash-additive operator \(\rho_i^{x_i}\) satisfying
	\[
	\rho_i^{x_i}(v+c\1)=\rho_i^{x_i}(v)+c
	\qquad
	\text{for every }v\in\R^{S_i},\ c\in\R .
	\]
	
	Then the elapsed-calendar constrained contingent multi-schedule problem is a rested generalized multi-armed bandit. Its value \(V:\mathbf S\to\R\) is characterized by the Bellman equation
	\begin{equation}\label{eq:bandit-bellman}
		V(x)
		=
		\max_{1\le i\le N}
		\left\{
		u(x_i)
		+
		\beta
		\rho_i^{x_i}
		\left(
		\bigl(V(x_{-i},y_i)\bigr)_{y_i\in S_i}
		\right)
		\right\}.
	\end{equation}
	
	For each scalar shadow price \(\lambda\in\R\) of elapsed calendar time, define the one-schedule value \(V_i^\lambda:S_i\to\R\) by
	\begin{equation}\label{eq:charged-stopping-value}
		V_i^\lambda(x_i)
		=
		\max\left\{
		\frac{\lambda}{1-\beta},
		\,
		u(x_i)
		+
		\beta
		\rho_i^{x_i}
		\left(
		\bigl(V_i^\lambda(y_i)\bigr)_{y_i\in S_i}
		\right)
		\right\}.
	\end{equation}
	Let
	\[
	D_i(\lambda)
	:=
	\left\{
	x_i\in S_i:
	V_i^\lambda(x_i)=\frac{\lambda}{1-\beta}
	\right\}
	\]
	be the set of local states in which schedule \(i\) is put on hold at shadow price \(\lambda\), and define the index
	\begin{equation}\label{eq:index-definition}
		\Lambda_i(x_i)
		:=
		\inf\{\lambda\in\R:x_i\in D_i(\lambda)\}.
	\end{equation}
	The on-hold sets are nested:
	\[
	\lambda'\ge\lambda
	\quad\Longrightarrow\quad
	D_i(\lambda)\subseteq D_i(\lambda')
	\qquad\text{for every }i.
	\]
	Hence \(\Lambda_i(x_i)\) is well defined. An optimal policy for the elapsed-calendar constrained contingent multi-schedule problem is to advance, at each product state \(x\), any schedule in
	\begin{equation}\label{eq:max-index-policy}
		\argmax_{1\le i\le N}\Lambda_i(x_i).
	\end{equation}
\end{theorem}

The first part of the theorem is the bandit representation. At a product state \(x\), the only calendar decision is which local clock to advance. If schedule \(i\) is advanced, only the \(i\)th coordinate moves. The continuation values relevant to that local schedule are therefore
\[
\bigl(V(x_{-i},y_i)\bigr)_{y_i\in S_i},
\]
which is exactly the vector appearing inside \(\rho_i^{x_i}\) in \eqref{eq:bandit-bellman}. The elapsed-calendar constraint therefore turns a collection of local stopping problems into a rested product problem.

The second part identifies the index. The index is not an additional priority ranking imposed from outside the model. It is computed from the one-schedule problem. At shadow price \(\lambda\), putting schedule \(i\) on hold indefinitely has value \(\lambda/(1-\beta)\). The set \(D_i(\lambda)\) consists of the local states at which that on-hold value is at least as good as advancing the schedule one more local step. The index \(\Lambda_i(x_i)\) is the critical shadow price at which putting the schedule on hold first becomes optimal in this isolated one-schedule comparison.

This language should be read as part of the one-schedule shadow-price calculation, not as literal abandonment of a real-world responsibility. A schedule that is not advanced at a given calendar date has not disappeared. In the constrained multi-schedule problem, it remains at its current local state while another schedule's local clock advances. The one-schedule calculation asks a sharper question: how high must the shadow price of calendar time be before advancing this schedule is no longer worthwhile? The answer is the index.

This is the shadow-price interpretation of the calendar-budget analogy. In a consumer problem, a Lagrange multiplier prices the scarce resource in the budget constraint. Here, \(\lambda\) prices one unit of elapsed calendar time. The one-schedule problem asks whether schedule \(i\), at local state \(x_i\), is still worth advancing when calendar time has shadow price \(\lambda\). The max-index rule advances the schedule with the highest current critical shadow price.

Cash additivity is what makes this scalar shadow price legitimate. To see this, subtract the on-hold value from the one-schedule problem. Define
\[
F_i^\lambda(x_i)
:=
V_i^\lambda(x_i)-\frac{\lambda}{1-\beta}.
\]
Using cash additivity of \(\rho_i^{x_i}\), \eqref{eq:charged-stopping-value} is equivalent to
\begin{equation}\label{eq:charged-excess-value}
	F_i^\lambda(x_i)
	=
	\max\left\{
	0,
	\,
	u(x_i)-\lambda
	+
	\beta
	\rho_i^{x_i}
	\left(
	\bigl(F_i^\lambda(y_i)\bigr)_{y_i\in S_i}
	\right)
	\right\}.
\end{equation}
In this excess-value form, \(\lambda\) is the shadow price of one unit of elapsed calendar time. The on-hold branch is normalized to zero. Increasing \(\lambda\) lowers the active branch by the same scalar in every local state. Since the on-hold branch is fixed, the set of states in which holding the schedule fixed is optimal expands monotonically with \(\lambda\). This is the indexability statement in the theorem.

The connection with the usual Weber and Whittle arguments is direct. Weber's charge formulation prices continuation directly: after subtracting the common on-hold value, the scalar appears as the price of advancing the schedule. Whittle's passive-subsidy formulation uses the opposite side of the same comparison: a rested inactive schedule receives the per-period subsidy \(\lambda\), which is equivalent to the one-time value \(\lambda/(1-\beta)\). In either sign convention, the scalar is the Lagrange multiplier, or shadow price, associated with elapsed calendar time. R7 is the primitive condition ensuring that this scalar can be moved through the recursive continuation aggregator without being distorted.

To get an idea about how the proof works, notice that the product-state Bellman equation follows from the elapsed-calendar constraint: a feasible calendar step advances exactly one local clock.  Maximizing over the schedule to be advanced gives \eqref{eq:bandit-bellman}. Monotonicity and cash additivity imply that each \(\rho_i^{x_i}\) is nonexpansive in the sup norm, so the discounted Bellman operator is a contraction.

For the index part, introduce a Lagrangian relaxation of the calendar constraint. In the relaxed problem, schedules choose active/passive independently rather than being required path by path to have exactly one active schedule each date; the scarcity of calendar time is priced by a scalar multiplier \(\lambda\). Cash additivity is what lets this scalar price enter each one-schedule problem as in \eqref{eq:charged-excess-value}, so the relaxed problem separates schedule by schedule. The key point is that a relaxed optimum is not generally feasible for the original problem: it may advance several schedules, or none. The index argument shows why the relaxation is tight here. A suitable tie-breaking rule selects a relaxed optimizer that advances exactly one maximal-index schedule and holds all nonmaximal schedules fixed. This selected relaxed optimizer satisfies the original elapsed-calendar constraint. Since the relaxation upper-bounds the original problem, and the max-index implementation attains that relaxed bound while satisfying the original constraint, the max-index policy is optimal.

The result is deliberately one-directional. The primitive schedule axioms, common-tail invariance, and rested product construction deliver a generalized bandit representation and its index policy. The theorem does not claim that every abstract bandit representation has a primitive schedule foundation of this form. In the special cases below, stronger equivalence statements are available because additional axioms pin down the continuation aggregator more sharply.

\section{Extensions and Special Cases}\label{sec:extensions}

Now we show familiar models arise through additional restrictions/axioms on our schedule domain. The first is the subjective expected-utility case, where the abstract continuation aggregator becomes a subjective transition kernel and the index reduces to the classical Gittins--Jones ratio. The second is variational preferences, which nest max-min as a further special case. The third is the rank-dependent, or Choquet, case, where continuation branches are evaluated by their rank rather than by additive probabilities. The final one is Pandora's problem, which is not a different continuation aggregator but a restriction on the local schedule domain: one local step opens a box and resolves all uncertainty.

For each state \(x\), define the induced continuation preference \(\succeq_x^C\) on continuation families by
\begin{equation}\label{eq:cont-pref}
	(a^y)_{y\in S}\succeq_x^C(b^y)_{y\in S}
	\quad\Longleftrightarrow\quad
	\Gamma_x(0,(a^y)_{y\in S})\succeq_x\Gamma_x(0,(b^y)_{y\in S}).
\end{equation}
Thus the current consequence is fixed at the null lottery, and only the contingent continuation tree varies. By R1, this comparison depends only on the continuation equivalence classes. Under D+R, it is represented by
\[
(a^y)_y
\mapsto
\rho_x\bigl((U_y(a^y))_{y\in S}\bigr).
\]
In this section, the continuation domain is assumed to contain the mixtures and common shifts used in the stated axioms. For notational purposes, continuation families can be identified with their continuation-value vectors
\[
U(a):=(U_y(a^y))_{y\in S}\in\R^S.
\]
Let \(\mathfrak T_x^+\) denote the finite stopping times \(\tau\ge1\) for the one-schedule problem starting from \(x\). Under the consequence-bearing-state convention, the cash-additive shadow-price equation is
\begin{equation}\label{eq:special-excess-template}
	F^\lambda(x)
	=
	\max\left\{
	0,\,
	u(x)-\lambda+\beta\rho_x(F^\lambda)
	\right\}.
\end{equation}
The associated per-period index is
\begin{equation}\label{eq:special-index-template}
	I(x):=\inf\{\lambda\in\R:F^\lambda(x)=0\}.
\end{equation}
The role of each special case is to identify the relevant continuation aggregator \(\rho_x\), and then substitute it into \eqref{eq:special-excess-template}.

\subsection{Subjective expected utility}

The subjective expected-utility case is the first special case we consider. Many familiar economic bandit models, including Bayesian experimentation and learning models, fall into this framework \citep{Rothschild1974,BoltonHarris1999,KellerRadyCripps2005,BergemannValimaki2008}. The restriction is not imposed directly on a transition law. It is imposed on continuation schedules. Once the current consequence is fixed, the decision maker satisfies ordinary independence over objective mixtures of continuation trees. Under that restriction, the continuation aggregator is linear, so uncertainty in a schedule is summarized by a subjective transition kernel.

This case is sharper than the general bandit representation above. In the general result, D, R, and R7 are enough to obtain a recursive stopping problem, scalar shadow prices, and an index rule. Under subjective expected utility, an additional continuation-independence axiom pins down the continuation aggregator as a linear expectation. Thus the primitive schedule restrictions become \emph{equivalent} to the standard expected-utility bandit representation.

\begin{assumption}[EU: one-step continuation independence]\label{ass:eu-cont-ind}
	For every \(x\in S\), the continuation domain is closed under objective statewise mixtures.\footnote{Mixtures are formed componentwise across next states. If one component schedule has already stopped, it contributes the null lottery \(0\) at all remaining local dates until the other component also stops; if both have stopped, the mixed component is stopped.} For all continuation families \(a,b,c\) and all \(\alpha\in(0,1)\),
	\begin{equation}\label{eq:eu-cont-ind}
		a\succeq_x^C b
		\quad\Longleftrightarrow\quad
		\alpha a+(1-\alpha)c\succeq_x^C \alpha b+(1-\alpha)c.
	\end{equation}
\end{assumption}

The axiom is the Anscombe--Aumann independence condition applied to continuation schedules. Once the current consequence is fixed, mixing both continuation trees with the same background continuation tree cannot change their ranking. The representation therefore recovers a subjective law over next local states.

\begin{proposition}[Subjective expected utility and the Gittins--Jones index]\label{prop:eu-index}
	Assume D, R and assumption \ref{ass:eu-cont-ind}. Then the following are equivalent.
	\begin{enumerate}[label=(\roman*)]
		\item EU holds for every \(x\in S\).
		\item There exists a Markov kernel \(p=(p_x)_{x\in S}\), with \(p_x\in\Delta(S)\), such that
		\begin{equation}\label{eq:eu-linear-rho-special}
			\rho_x(v)=\sum_{y\in S}p_x(y)v_y
			\qquad\text{for every }x\in S.
		\end{equation}
		Equivalently,
		\begin{equation}\label{eq:eu-recursive-special}
			U_x\Bigl(\Gamma_x(\ell,(a^y)_{y\in S})\Bigr)
			=
			u(\ell)+\beta\sum_{y\in S}p_x(y)U_y(a^y).
		\end{equation}
	\end{enumerate}
	Under either condition, R7 holds. The shadow-price equation is
	\begin{equation}\label{eq:eu-excess-special}
		F^\lambda(x)
		=
		\max\left\{
		0,\,
		u(x)-\lambda+\beta\sum_{y\in S}p_x(y)F^\lambda(y)
		\right\},
	\end{equation}
	and the resulting index is the Gittins--Jones ratio
	\begin{equation}\label{eq:eu-index-special}
		I^{EU}(x)
		=
		\sup_{\tau\in\mathfrak T_x^+}
		\frac{
			\E_x^p\left[\sum_{t=0}^{\tau-1}\beta^t u(X_t)\right]
		}{
			\E_x^p\left[\sum_{t=0}^{\tau-1}\beta^t\right]
		}.
	\end{equation}
\end{proposition}

\subsubsection{Finite good-news learning}\label{subsubsec:learning}

Learning is not a new continuation aggregator. It is a special case of the
subjective expected-utility representation in which the local state of a strand
records the progress of information. The primitive object is still a schedule
whose local states carry consequences. The additional structure is a finite
schedule subdomain and preference restrictions that make the represented SEU kernel
look like a familiar learning process.

We spell this out for a stationary good-news model. This is only one example. The
same logic can accommodate other finite learning structures by changing the finite
schedule subdomain and the corresponding consistency restrictions. The point of
using integer states below is deliberate: the state space need not be written as a
set of posterior probabilities or likelihood ratios. For any fixed finite horizon,
the reachable learning states are finite. We can label them by integers, or by any
other finite ordered set. The posterior and log-likelihood scale are recovered
after the representation.

For a finite horizon \(T\), consider the finite learning subdomain
\[
	S_T^G=\{\gamma,0,1,\ldots,T,\rho\}\subseteq S .
\]
The integer states \(0,1,\ldots,T\) are unresolved learning states. The state
\(\rho\) is the resolved-good state, and \(\gamma\) is a lower benchmark state.
The active unresolved states are \(1,\ldots,T\). State \(0\) is the terminal
unresolved state for this finite truncation. The primitive no-news successor of an
active state \(n\) is \(n-1\). These are only labels. In particular, the integer
distance between \(n\) and \(n-1\) is not assumed to be a unit of log likelihood.

For an active integer state \(n\), write \([z]_n\) for the constant continuation
family at \(n\) that delivers the consequence carried by state \(z\) at every
next-state coordinate:
\[
	[z]_n^y=d_y((z)),
	\qquad y\in S_T^G .
\]
Here \(d_y((z))\) denotes the one-period deterministic schedule from state \(y\)
that delivers the consequence carried by state \(z\). For an event
\(A\subseteq S_T^G\), write
\[
	[\rho\text{ on }A;\gamma\text{ off }A]_n
\]
for the continuation family at \(n\) that delivers the consequence carried by
\(\rho\) on coordinates in \(A\), and the consequence carried by \(\gamma\) on
coordinates outside \(A\). More generally,
\[
	[\rho\text{ on }\{\rho\};\, n-1\text{ on }\{n-1\};\,\gamma
	\text{ off }\{\rho,n-1\}]_n
\]
denotes the continuation family that delivers \(\rho\) on the \(\rho\)-coordinate,
\(n-1\) on the \(n-1\)-coordinate, and \(\gamma\) on every other coordinate.
Mixtures of such continuation families are the objective statewise mixtures used
in Assumption~\ref{ass:eu-cont-ind}.

\begin{assumption}[GNL: finite good-news learning]\label{ass:good-news-learning}
	On the finite learning subdomain
	\[
		S_T^G=\{\gamma,0,1,\ldots,T,\rho\},
	\]
	the following restrictions hold.

	\begin{enumerate}[label=(\roman*)]
		\item \emph{State ordering.} For every active integer state
		\(n=1,\ldots,T\) and every unresolved state \(j=0,\ldots,T\),
		\[
			[\rho]_n\succ_n^C [j]_n\succ_n^C[\gamma]_n .
		\]
		Moreover, for every active integer state \(n=1,\ldots,T\),
		\[
			[n]_n
			\succ_n^C
			[\rho\text{ on }\{\rho\};\gamma\text{ off }\{\rho\}]_n
			\succ_n^C
			[\gamma]_n .
		\]

		\item \emph{Good-news/no-news exhaustiveness.} For every active integer
		state \(n=1,\ldots,T\),
		\begin{equation}\label{eq:gnl-exhaustiveness}
			[\rho\text{ on }\{\rho,n-1\};\gamma\text{ off }\{\rho,n-1\}]_n
			\sim_n^C
			[\rho]_n .
		\end{equation}

		\item \emph{Stationary relative value of good news.} For every pair of
		active integer states \(m,n=1,\ldots,T\) and every \(\alpha\in(0,1)\),
		\begin{equation}\label{eq:gnl-relative-stationarity}
			[\rho\text{ on }\{\rho\};\gamma\text{ off }\{\rho\}]_n
			\succeq_n^C
			\alpha[n]_n+(1-\alpha)[\gamma]_n
		\end{equation}
		if and only if
		\begin{equation*}
			[\rho\text{ on }\{\rho\};\gamma\text{ off }\{\rho\}]_m
			\succeq_m^C
			\alpha[m]_m+(1-\alpha)[\gamma]_m .
		\end{equation*}

		\item \emph{Successor-consequence consistency.} For every active integer
		state \(n=1,\ldots,T\),
		\begin{equation}\label{eq:gnl-successor-consistency}
			[n]_n
			\sim_n^C
			[\rho\text{ on }\{\rho\};\, n-1\text{ on }\{n-1\};\, \gamma
			\text{ off }\{\rho,n-1\}]_n .
		\end{equation}
	\end{enumerate}
\end{assumption}

The first clause gives the ordinal role of the special states. The resolved state
\(\rho\) is better than every unresolved integer state, and every unresolved
integer state is better than the lower benchmark \(\gamma\). The additional strict
comparison says that the good-news branch is valuable, but not already as valuable
as the current unresolved state. This rules out degenerate cases in which good
news is either worthless or already certain.

The second clause is an exhaustiveness test. At active state \(n\), a continuation
folder that pays \(\rho\) on the event \(\{\rho,n-1\}\) and \(\gamma\) otherwise is
equivalent to receiving \(\rho\) for sure. Under the SEU representation this will
mean that, from \(n\), all represented one-step weight lies on the two coordinates
\(\rho\) and \(n-1\). The axiom itself is only a preference comparison between
continuation folders.

The third clause is the stationarity restriction. It does not say that a primitive
probability is constant. It says that the branch bet on the resolved-good
coordinate has the same relative value at every active integer state. The
comparison is measured against mixtures of the current state \(n\) and the lower
benchmark \(\gamma\), not against an exogenous posterior label. The common relative
value will become the good-news intensity only after the SEU representation is
applied.

The fourth clause connects the integer labels to learning. At state \(n\), the
current unresolved consequence is equivalent to the one-step
successor-consequence folder: receive \(\rho\) on the resolved-good coordinate,
receive \(n-1\) on the no-news coordinate, and receive \(\gamma\) elsewhere. This
is the schedule version of a posterior martingale restriction, but it is stated
only as an indifference between continuation folders.

\begin{lemma}[Finite good-news representation]\label{lem:finite-good-news-representation}
	Assume D, R, EU, and GNL on
	\[
		S_T^G=\{\gamma,0,1,\ldots,T,\rho\}.
	\]
	Define the normalized revealed value of an unresolved integer state by
	\[
		\mu_n:=
		\frac{u(n)-u(\gamma)}{u(\rho)-u(\gamma)},
		\qquad n=0,\ldots,T .
	\]
	Then \(\mu_n\in(0,1)\). For every active integer state
	\(n=1,\ldots,T\), let \(q_n\in\Delta(S_T^G)\) be the represented one-step
	continuation law at \(n\). Then
	\[
		q_n(\{\rho,n-1\})=1.
	\]
	Moreover, there exists a constant \(\theta\in(0,1)\), independent of \(n\),
	such that
	\[
		q_n(\rho)=\theta\mu_n,
		\qquad
		q_n(n-1)=1-\theta\mu_n .
	\]
	The normalized revealed values satisfy
	\begin{equation}\label{eq:gnl-no-news-update}
		\mu_{n-1}
		=
		\frac{(1-\theta)\mu_n}{1-\theta\mu_n}.
	\end{equation}
	Equivalently,
	\begin{equation}\label{eq:gnl-log-odds-update}
		\log\frac{\mu_{n-1}}{1-\mu_{n-1}}
		=
		\log\frac{\mu_n}{1-\mu_n}
		+
		\log(1-\theta).
	\end{equation}
	Hence the integer states are identified with an affine log-likelihood scale:
	\[
		r_n=a+n\delta,
		\qquad
		\delta:=-\log(1-\theta).
	\]
	The represented process is the finite stationary good-news process: good news
	arrives with probability \(\theta\) conditional on the good type and never
	arrives conditional on the bad type.
\end{lemma}

The lemma shows how the finite integer schedule can be read as good-news learning.
Before the representation, the states are just
\[
	\gamma,0,1,\ldots,T,\rho .
\]
After D, R, and EU, each active integer state has a represented one-step law over
next-state coordinates. GNL disciplines those represented laws. Exhaustiveness
puts all weight on the resolved-good coordinate \(\rho\) and the one-step no-news
coordinate \(n-1\). Stationary relative valuation makes the represented value of
the \(\rho\)-branch proportional to the current revealed value \(\mu_n\), with the
same factor \(\theta\) at every active integer state. Successor-consequence
consistency then forces the no-news successor values to satisfy
\[
	\frac{\mu_{n-1}}{1-\mu_{n-1}}
	=
	(1-\theta)\frac{\mu_n}{1-\mu_n}.
\]
Thus the integer labels are not assumed to be log-likelihoods. They become
log-likelihood states in the representation, with one step down the integer chain
corresponding to the constant no-news decrement
\[
	-\log(1-\theta).
\]

This construction illustrates the general point. A learning model can be encoded
by a finite schedule subdomain whose states are ordinary consequence-bearing
states. The posterior, likelihood ratio, and signal probabilities are then
features of the represented SEU kernel, not primitives of the domain. We used
integers because they make this separation transparent, but the same finite chain
could be relabeled by rationals, names, or any other finite set. What matters is
the preference structure that makes the represented process isomorphic to the
standard good-news experiment.

No additional index theorem is required. Once the represented SEU kernel is
obtained, the learning strand is governed by the same Gittins--Jones index in
Proposition~\ref{prop:eu-index}. The learning restriction only explains when the
kernel behind that index is the one generated by good-news-only Bayesian updating. Extending to infinite horizons also follows by simply taking limits.

\subsection{Variational preferences}

Expected utility treats a continuation folder as if there were a single subjective law for the next local state. Many responsibility strands are not naturally like that. A work file may depend on how an editor or referee will react; a school matter may depend on how a teacher interprets a child's performance; a household repair may turn out to be minor, serious, or entangled with another task. The person may rank these continuation folders without being willing to name one transition law as the right description of the strand.

Variational preferences, with max-min expected utility as a special case, capture this kind of ambiguous scheduling. The decision maker evaluates a continuation folder against possible one-step descriptions of the next local state. A law \(p\) may be used to summarize the branch that will be reached, but using that law carries a penalty \(\alpha_x(p)\). Low-penalty laws are treated as plausible descriptions of the strand; high-penalty laws remain possible but are costly to use in the evaluation. This differs from robust bandit models that place ambiguity on the full multi-project environment, where an adverse model may choose a joint description of all projects and thereby link alternatives that would otherwise be separate \citep{KimLim2016,CaroGupta2022}. In the schedule interpretation, a work obligation, a household repair, a medical follow-up, and a family responsibility may all compete for the same calendar, but uncertainty about one strand need not be tied to uncertainty about the others. We therefore allow rich ambiguity and robustness within each local continuation problem, while keeping the cross-strand interaction exactly where the model puts it: in the scarce calendar.

\begin{assumption}[VA: one-step continuation uncertainty aversion]\label{ass:cont-unc-av}
	For every \(x\in S\), the continuation domain is closed under objective statewise mixtures, and for all continuation families \(a,b\),
	\[
	a\sim_x^C b
	\quad\Longrightarrow\quad
	\alpha a+(1-\alpha)b\succeq_x^C a
	\qquad\forall\alpha\in(0,1).
	\]
\end{assumption}

The axiom is the continuation-schedule version of uncertainty aversion. If two continuation folders are equally good, then objective randomization between them cannot make the decision maker worse off; mixing smooths the exposure to any single ambiguous continuation folder. Mixtures are formed state by state, using the same stopped-schedule convention as in Assumption~\ref{ass:eu-cont-ind}.

For the dynamic formulas below, a path law is a complete scenario book for the future local evolution of the strand. Given a finite history \(h_t=(x_0,\ldots,x_t)\), write \(p_t^{\mathbb P}(\cdot\mid h_t)\) for the one-step law selected by \(\mathbb P\) after that history. Let
\begin{equation}\label{eq:variational-path-laws-special}
\begin{aligned}
	\mathcal P_x^\Delta
	:=
	\Bigl\{\mathbb P:\ &\mathbb P(X_0=x)=1,
	\text{ and there are kernels }p_t^{\mathbb P}(\cdot\mid h_t)\in\Delta(S)\\
	&\text{such that }\mathbb P(X_{t+1}\in\cdot\mid H_t=h_t)=p_t^{\mathbb P}(\cdot\mid h_t)
	\text{ for every }t,h_t
	\Bigr\}.
\end{aligned}
\end{equation}
where \(H_t=(X_0,\ldots,X_t)\). Thus the infimum over \(\mathbb P\) below is an infimum over rectangular, history-dependent one-step descriptions of how the responsibility strand may branch. Just like subjective utility, assumption \ref{ass:cont-unc-av} disciplines the aggregator so that the representation is once again equivalent.

\begin{proposition}[Variational continuation and the penalized index]\label{prop:variational-index}
	Assume D, R and assumption \ref{ass:cont-unc-av}. Then the following are equivalent.
	\begin{enumerate}[label=(\roman*)]
		\item R7 and VA hold for every \(x\in S\).
		\item For each \(x\), there exists a convex lower-semicontinuous penalty \(\alpha_x:\Delta(S)\to[0,\infty]\), normalized by \(\min_p\alpha_x(p)=0\), such that
		\begin{equation}\label{eq:variational-rho-special}
			\rho_x(v)=\min_{p\in\Delta(S)}\{p\cdot v+\alpha_x(p)\}.
		\end{equation}
	\end{enumerate}
	Under either condition, the charged stopping equation is
	\begin{equation}\label{eq:variational-excess-special}
		F^W(x)
		=
		\max\left\{
		0,\,
		u(x)-W+\beta\min_{p\in\Delta(S)}
		\bigl[p\cdot F^W+\alpha_x(p)\bigr]
		\right\}.
	\end{equation}
	The resulting variational index is
	\begin{equation}\label{eq:variational-index-special}
		I^{VAR}(x)
		=
		\sup_{\tau\in\mathfrak T_x^+}\inf_{\mathbb P}
		\frac{
			\E_{\mathbb P}\left[
			\sum_{t=0}^{\tau-1}\beta^t u(X_t)
			+\sum_{t=0}^{\tau-1}\beta^{t+1}\alpha_{X_t}(p_t)
			\right]
		}{
			\E_{\mathbb P}\left[\sum_{t=0}^{\tau-1}\beta^t\right]
		},
	\end{equation}
	where \(p_t\) is the one-step transition law selected at history \(t\).
\end{proposition}

Max--min is the stark robust version of the variational story. At each local state \(x\), the decision maker keeps a set \(P(x)\) of plausible one-step stories and evaluates the continuation folder by the least favorable one. Laws inside \(P(x)\) receive zero penalty; laws outside \(P(x)\) are ruled out by an infinite penalty. The rectangularity in the index means that after each future local history the relevant ambiguity set is determined by the current local state of the same strand.

Given a correspondence \(P:S\rightrightarrows\Delta(S)\), define the associated rectangular path-law family by
\begin{equation}\label{eq:maxmin-path-laws-special}
	\mathcal P_x^P
	:=
	\left\{
	\mathbb P\in\mathcal P_x^\Delta:
	p_t^{\mathbb P}(\cdot\mid H_t)\in P(X_t)
	\text{ for every }t\text{, }\mathbb P\text{-a.s.}
	\right\}.
\end{equation}

\begin{corollary}[Rectangular max--min as a special case]\label{cor:maxmin-index}
	Assume D and R, and let \(P(x)\subseteq\Delta(S)\) be nonempty, closed, and convex for every \(x\in S\). In Proposition~\ref{prop:variational-index}, set
	\begin{equation}\label{eq:maxmin-indicator-penalty}
		\alpha_x^P(p)
		=
		\begin{cases}
		0, & p\in P(x),\\
		+\infty, & p\notin P(x).
		\end{cases}
	\end{equation}
	Then
	\begin{equation}\label{eq:maxmin-rho-special}
		\rho_x(v)=\min_{p\in P(x)}p\cdot v.
	\end{equation}
	The charged stopping equation is
	\begin{equation}\label{eq:maxmin-excess-special}
		F^W(x)
		=
		\max\left\{
		0,\,
		u(x)-W+\beta\min_{p\in P(x)}p\cdot F^W
		\right\}.
	\end{equation}
	The resulting robust index is
	\begin{equation}\label{eq:maxmin-index-special}
		I^{MM}(x)
		=
		\sup_{\tau\in\mathfrak T_x^+}
		\inf_{\mathbb P\in\mathcal P_x^P}
		\frac{
			\E_{\mathbb P}\left[\sum_{t=0}^{\tau-1}\beta^t u(X_t)\right]
		}{
			\E_{\mathbb P}\left[\sum_{t=0}^{\tau-1}\beta^t\right]
		}.
	\end{equation}
\end{corollary}

\subsection{Rank-dependent and Choquet preferences}

Expected utility treats continuation branches additively. Rank-dependent preferences allow the decision maker to evaluate the continuation tree by the rank of its branches. This is useful in the schedule domain because a strand of responsibility often has salient upside or downside branches. A work strand may become important if the feedback is promising; a child-related strand may become urgent if a report card reveals a problem; a household strand may become costly if a repair turns out to be serious. The decision maker need not weight these branches as if only their probabilities matter.

In this sense, what looks like loss aversion in a static lottery becomes something closer to Fear of Missing Out (FOMO) in a continuation problem. The salient object is not merely a low prize relative to a reference point, but the possibility of missing a valuable continuation branch by putting the schedule on hold. A rank-dependent continuation aggregator can make the best or worst branches loom larger than they would under additive probabilities.

For value vectors \(v,w\in\R^S\), call \(v\) and \(w\) comonotonic if there are no states \(y,z\) such that
\[
(v_y-v_z)(w_y-w_z)<0.
\]
Continuation families are comonotonic when their induced value vectors are comonotonic.

\begin{assumption}[CEU: one-step comonotonic continuation independence]\label{ass:ceu-cont-ind}
	For every \(x\in S\), the continuation domain is closed under objective statewise mixtures, and for all pairwise comonotonic continuation families \(a,b,c\) and all \(\alpha\in(0,1)\),
	\begin{equation}\label{eq:ceu-cont-ind}
		a\succeq_x^C b
		\quad\Longleftrightarrow\quad
		\alpha a+(1-\alpha)c\succeq_x^C \alpha b+(1-\alpha)c.
	\end{equation}
\end{assumption}

CEU keeps the independence logic only when the continuation trees agree on the ranking of states. This is the primitive schedule restriction behind rank-dependent evaluation of continuation uncertainty. Once again we see that the discipline imposed by the assumption is enough to strengthen the bandit representation to equivalence.

\begin{proposition}[Choquet continuation and the rank-dependent index]\label{prop:choquet-index}
	Assume D, R and assumption \ref{ass:ceu-cont-ind}. Then the following are equivalent.
	\begin{enumerate}[label=(\roman*)]
		\item CEU holds for every \(x\in S\).
		\item For each \(x\), there exists a normalized capacity \(\nu_x:2^S\to[0,1]\) such that
		\[
		\nu_x(\varnothing)=0,\qquad \nu_x(S)=1,\qquad A\subseteq B\Rightarrow \nu_x(A)\le\nu_x(B),
		\]
		and
		\begin{equation}\label{eq:choquet-rho-special}
			\rho_x(v)=\int v\,d\nu_x,
		\end{equation}
		where the integral is the finite-state Choquet integral.
	\end{enumerate}
	Under either condition, R7 holds. The shadow-price equation is
	\begin{equation}\label{eq:choquet-excess-special}
		F^\lambda(x)
		=
		\max\left\{
		0,\,
		u(x)-\lambda+\beta\int F^\lambda\,d\nu_x
		\right\}.
	\end{equation}
	The resulting rank-dependent index is
	\begin{equation}\label{eq:choquet-index-special}
		I^{CH}(x):=\inf\{\lambda:F^\lambda(x)=0\}.
	\end{equation}
	If \(\nu_x\) is additive, this reduces to the subjective expected-utility case.
\end{proposition}

\subsection{Pandora boxes}

Pandora boxes are different from the preceding cases. Expected utility, variational preferences, and Choquet preferences restrict how a continuation tree is evaluated. Pandora's problem restricts what the local schedule can look like. Each box has an unopened state. One local step opens the box and resolves all uncertainty. After that, the revealed state is absorbing. This is the sense in which the framework is complementary to the Pandora search literature: reservation values arise here from a domain restriction on local schedules, rather than from taking boxes, inspection costs, and reservation equations as primitive \citep{Weitzman1979,OlszewskiWeber2015,Doval2018,AusterChe2025}.

\begin{assumption}[PB: absorbing one-step revelation]\label{ass:pandora-domain}
	For each box \(i\), the local state space is
	\begin{equation}\label{eq:pandora-local-state-special}
		S_i^P:=\{q_i\}\cup Z_i,
	\end{equation}
	where \(q_i\) is the unopened state and \(Z_i\) is the finite set of revealed states. The local histories satisfy
	\[
	\Omega_{i,q_i}^P=\{(q_i,z,z,\ldots):z\in Z_i\},
	\qquad
	\Omega_{i,z}^P=\{(z,z,\ldots)\}.
	\]
	Opening box \(i\) has one-period utility cost \(c_i\ge0\). Each revealed state \(z\in Z_i\) is absorbing and has one-period utility \(u(z)\). Its absorbing value is
	\begin{equation}\label{eq:pandora-flow-value-special}
		g_i(z):=\frac{u(z)}{1-\beta}.
	\end{equation}
\end{assumption}

The primitive content is that all uncertainty in a box is resolved by the first local step. Once the box is revealed, the associated schedule has no further branching.

\begin{proposition}[Pandora reservation indices]\label{prop:pandora-index}
	Assume D, R, and R7 on the restricted Pandora local domain, and impose PB. For each unopened box \(i\), there is a unique reservation value \(\pi_i\) satisfying
	\begin{equation}\label{eq:pandora-reservation}
		\pi_i
		=
		-c_i+\beta\rho_i\bigl((\pi_i\vee g_i(z))_{z\in Z_i}\bigr).
	\end{equation}
	Equivalently, by cash additivity,
	\begin{equation}\label{eq:pandora-reservation-cash}
		(1-\beta)\pi_i+c_i
		=
		\beta\rho_i\bigl(((g_i(z)-\pi_i)^+)_{z\in Z_i}\bigr).
	\end{equation}
	The per-period indices on the Pandora local domain are
	\begin{equation}\label{eq:pandora-index-special}
		I_i^P(q_i)=(1-\beta)\pi_i,
		\qquad
		I_i^P(z)=u(z)=(1-\beta)g_i(z).
	\end{equation}
	Equivalently, in one-time value units, the max-index rule opens a box with maximal reservation value \(\pi_i\) whenever that value exceeds the best currently revealed absorbing value, and otherwise commits to the best revealed absorbing state.
	
	In the subjective expected-utility specialization
	\[
	\rho_i(v)=\sum_{z\in Z_i}p_i(z)v(z),
	\]
	\eqref{eq:pandora-reservation-cash} becomes
	\begin{equation}\label{eq:pandora-seu-reservation-special}
		(1-\beta)\pi_i+c_i
		=
		\beta\sum_{z\in Z_i}p_i(z)(g_i(z)-\pi_i)^+.
	\end{equation}
	Under the classical normalization \(g_i(z)=z\), the formal undiscounted boundary case \(\beta=1\) gives Weitzman's reservation equation
	\[
	c_i=\sum_{z\in Z_i}p_i(z)(z-\pi_i)^+.
	\]
\end{proposition}

\section{Conclusion}\label{sec:conclusion}

Bandit models are usually introduced as collections of arms. Each arm has a state, an evolution rule, and a reward process; the decision maker then chooses which arm to operate. This paper has taken a different starting point. The primitive object is not an arm, but a stopped local contingent schedule: a possible unfolding of one responsibility, project, experiment, or opportunity in its own local time. The economic problem is then to place many such local schedules into one scarce calendar.

The change in domain is the main point of the paper. In consumer theory, one does not begin with demand; one begins with preferences over consumption bundles and derives demand by adding a budget constraint. Here, one does not begin with a bandit; one begins with preferences over stopped schedules and derives a bandit by adding an elapsed-calendar constraint. The analogy is substantive. Bundles are the objects that can be purchased or forgone. Local contingent schedules are the objects that can be advanced or postponed without disappearing.

The first representation theorem shows when preferences over one stopped schedule can be written as a generalized stopping problem. Advancing a strand gives a current consequence, future consequences are discounted in local time, and uncertainty about continuation branches is evaluated by an aggregator. Recursive representations are familiar in decision theory, but the stopped-schedule domain gives them a different role. On this domain, the usual behavioral restrictions do not merely justify a convenient dynamic formula; they characterize when preferences over possible unfoldings of a strand have the structure of a stopping problem.

The move from stopping to bandits requires one further behavioral restriction. Calendar time must be priced in units that remain meaningful across strands. The common-tail compensation axiom supplies this property. If two ways of beginning a strand are exactly balanced, then attaching the same downstream continuation behind both of them should not break the tie. This is what allows a scalar shadow price of calendar time to pass through the local continuation evaluation. Once this is true, the elapsed-calendar constraint turns the collection of stopped schedules into a rested bandit.

The index is therefore not an additional ranking rule imposed by the analyst. It is the critical shadow price at which a strand is just worth putting on hold. A high-index strand is one whose next local step remains worthwhile even when calendar time is expensive. The max-index rule is optimal because the relaxed calendar problem separates into one-strand stopping problems, and the policy selected by those separated problems already satisfies the original calendar constraint.

The special cases show how familiar models fit inside this foundation. Under expected utility, the continuation aggregator becomes a subjective transition law and the classical Gittins--Jones index is recovered. In learning problems, beliefs are local states: they move when a strand is advanced and remain fixed when it is rested. With variational, multiplier, or max--min preferences, the index becomes a robust priority price. With rank-dependent or Choquet preferences, salient continuation branches can loom large, producing a dynamic version of fear of missing out. Pandora's problem comes from a different source: a restriction on the schedule domain in which opening a box resolves all uncertainty in one local step.

The framework also clarifies what is not primitive. A transition law, ambiguity set, capacity, or reservation value is not the starting point. Each appears only after adding restrictions to preferences, consequences, or the schedule domain. Ambiguity attitudes, for example, can enter through the continuation aggregator, but they can also be introduced by enriching consequences themselves, as with horse lotteries. The point is not to privilege one formal treatment of uncertainty. It is to identify the schedule structure that makes stopping and prioritization meaningful before any particular representation of uncertainty is chosen.

Several extensions are natural. One is to allow unattended strands to evolve, depreciate, or expire, producing a behavioral foundation for restless bandits. Another is to replace the single calendar with several scarce inputs: attention, expertise, money, or physical capacity. A third is to study richer schedule domains in which strands arrive, merge, split, or create obligations for one another. These extensions would move beyond the rested, single-calendar environment studied here while preserving the same organizing principle: represent the evaluation of one local strand first, then impose the scarcity constraint that makes prioritization necessary.

The paper therefore gives a foundation for bandits as economic choice problems. It identifies the primitive alternatives, the behavioral content of stopping, and the scarcity constraint that turns stopping into prioritization. On this view, an arm is not a primitive machine. It is a represented responsibility strand. A bandit is not just a Markov control problem. It is a demand problem for local advancement under a scarce calendar. The index is the resulting shadow price of moving one strand forward.

\newpage
\appendix
\section{Proofs for the Recursive Stopping Representation}\label{app:recursive}

\begin{proof}[Proof of Proposition~\ref{prop:detstreams}]
	The argument is the standard finite-stream representation argument, adapted to the stopped-flow domain.
	
	Fix an initial state $x$. D1 gives a weak order on deterministic schedules, and D2 gives mixture independence on each fixed-horizon space $\Lcal^T$. Thus, for each $T$, there is an affine expected-utility representation of the restriction of $\succeq_x$ to $\Lcal^T$. D3 aligns the representation of a front block across common same-length tails. D5 then identifies streams that differ only by an appended terminal null consequence, so the weights cannot depend on an artificial terminal zero. D4 gives stationarity: delaying both streams by a common null consequence rescales the comparison but does not change its sign. Hence, after normalizing the first-period weight to one, the date-$t$ weight is $\beta^t$ for some $\beta\ge0$.
	
	D6 rules out both degeneracy and non-impatience. The first part of D6 gives lotteries $\ell,m$ with $u(\ell)>u(m)$. The second part gives
	\[
	u(\ell)+\beta u(m)>u(m)+\beta u(\ell),
	\]
	and therefore $\beta<1$. Since the second-period consequence matters in the deterministic representation, $\beta>0$. Thus $\beta\in(0,1)$. D8 makes the same one-period affine index and the same discount factor valid across initial states. D7 gives the required continuity. This yields
	\[
	U_x(d_x(\ell_0,\ldots,\ell_{T-1}))
	=
	\sum_{t=0}^{T-1}\beta^t u(\ell_t)
	\]
	for every $x$ and every finite deterministic stream. The usual vNM uniqueness argument gives uniqueness of $u$ up to positive affine transformations, with the corresponding normalization of the displayed representation.
\end{proof}

\begin{proof}[Proof of Theorem~\ref{thm:recursive}]
	By Proposition~\ref{prop:detstreams}, deterministic schedules are already represented on the common discounted utility scale. It remains to represent the one-step comparison between a current lottery and a vector of continuation schedules.
	
	Fix $x$. R1 implies that the one-step comparison depends on each continuation schedule only through its equivalence class under the relevant next-state preference. Hence the induced relation $\succeq_x^1$ on
	\[
	\Lcal\times Q,
	\qquad
	Q:=\prod_{y\in S}(\A_y/{\sim_y}),
	\]
	is well defined. R3 gives a continuous weak order on this one-step domain. R4 gives coordinate independence between the current lottery coordinate and the continuation-profile coordinate. R5 gives double cancellation, and R6 supplies the solvability needed to use the additive conjoint representation theorem on the attainable part of $\Lcal\times Q$. Therefore there exist numerical functions $A_x$ and $B_x$ such that
	\[
	(\ell,q)\succeq_x^1(m,q')
	\quad\Longleftrightarrow\quad
	A_x(\ell)+B_x(q)\ge A_x(m)+B_x(q').
	\]
	
	The restriction of this representation to deterministic two-date schedules pins down the current coordinate. Since the deterministic representation has already identified the affine one-period utility index, $A_x$ is a positive affine transform of $u$. We choose the common normalization in which $A_x=u$. The remaining coordinate can then be written as
	\[
	B_x(q)=\beta \rho_x(q),
	\]
	where $\rho_x$ is defined on the vector of continuation values represented by the coordinates of $q$. R2 implies that $\rho_x$ is increasing in the continuation values. Continuity of $\rho_x$ follows from R3 and the continuity of the one-step representation. This gives
	\[
	U_x\Bigl(\Gamma_x(\ell,(a^y)_{y\in S})\Bigr)
	=
	u(\ell)+\beta\rho_x\bigl(U_1(a^1),\ldots,U_n(a^n)\bigr),
	\]
	which is \eqref{eq:recursive-rep}.
	
	Conversely, suppose a representation of the displayed form exists. D follows from the restriction to deterministic schedules. R1 follows because only the continuation values enter the one-step representation. R2 follows from monotonicity of $\rho_x$. R3 follows from the numerical representation and continuity. R4 and R5 follow from the additive separation between the current coordinate and the continuation coordinate. R6 is exactly the interval-richness and solvability condition imposed on the attainable domain. Hence Assumptions~\ref{ass:D} and~\ref{ass:R} are equivalent to the recursive stopping representation.
\end{proof}

\begin{proof}[Proof of Corollary~\ref{cor:induced-stopping}]
	Fix an initial state \(x\). For a finite history
	\(h=(x_0,\ldots,x_t)\in H(x)\), write \(|h|=t\),
	\(\bar x(h)=x_t\), and
	\[
		h|_k:=(x_0,\ldots,x_k),\qquad k\le t.
	\]
	For \(y\in S\), write
	\[
		h,y:=(x_0,\ldots,x_t,y).
	\]
	
	A local stopping rule is equivalently described by its continuation set
	\(C\subseteq H(x)\), where \(h\in C\) means that the local schedule is advanced
	at history \(h\). The set \(C\) must be prehistory closed:
	\[
		h\in C,\ k\le |h|
		\quad\Longrightarrow\quad
		h|_k\in C.
	\]
	Equivalently, once a history is not in \(C\), no extension of that history is
	in \(C\). This is exactly absorbing nullity for the stopped schedule. Given
	such a \(C\), the associated stopped schedule assigns the consequence carried
	by \(\bar x(h)\) at histories \(h\in C\), and assigns \(\bot\) at histories
	\(h\notin C\).
	
	Under the consequence-bearing-state convention, the current consequence at
	history \(h\) is \(\bar x(h)\). Since \(S\) is finite, set
	\[
		\bar u:=\max_{z\in S}|u(z)|,
		\qquad
		B:=\frac{\bar u}{1-\beta}.
	\]
	On constant continuation vectors in the completed bounded range, the
	deterministic normalization gives
	\[
		\rho_z(c\mathbf 1)=c.
	\]
	
	Define finite-depth values \(V_x^n:H(x)\to\R_+\) by
	\[
		V_x^0(h)=0
		\qquad\text{for every }h\in H(x),
	\]
	and, for \(n\ge0\),
	\begin{equation}\label{eq:finite-envelope-proof}
		V_x^{n+1}(h)
		=
		\max\left\{
			0,\,
			u(\bar x(h))
			+
			\beta
			\rho_{\bar x(h)}
			\left(
				\bigl(V_x^n(h,y)\bigr)_{y\in S}
			\right)
		\right\}.
	\end{equation}
	We interpret \(V_x^n(h)\) as the value of the best continuation set in the
	subtree rooted at \(h\), subject to forced stopping after at most \(n\) further
	local advances.
	
	This interpretation follows by induction on \(n\). For \(n=0\), forced
	stopping gives value zero. Suppose the claim holds for \(n\). Starting from
	history \(h\), an \((n+1)\)-step rule either stops immediately, giving zero,
	or advances the schedule at \(h\). If it advances, the current contribution is
	\(u(\bar x(h))\), and after next local state \(y\) the remaining problem is an
	\(n\)-step problem rooted at \(h,y\). By the induction hypothesis, the
	branchwise optimal continuation values are \(V_x^n(h,y)\). Branchwise
	monotonicity of \(\rho_{\bar x(h)}\) then gives
	\eqref{eq:finite-envelope-proof}.
	
	The sequence \((V_x^n)\) is increasing. Since \(V_x^0\le V_x^1\), monotonicity
	of each \(\rho_z\) and of \(r\mapsto\max\{0,r\}\) gives
	\[
		V_x^n(h)\le V_x^{n+1}(h)
		\qquad\text{for every }n\text{ and every }h\in H(x).
	\]
	It is also bounded above by \(B\). If \(0\le V_x^n(h)\le B\) for every \(h\),
	then
	\[
		0\mathbf 1
		\le
		\bigl(V_x^n(h,y)\bigr)_{y\in S}
		\le
		B\mathbf 1 .
	\]
	By monotonicity and the constant-vector normalization,
	\[
		0
		\le
		\rho_{\bar x(h)}
		\left(
			\bigl(V_x^n(h,y)\bigr)_{y\in S}
		\right)
		\le
		B .
	\]
	Therefore
	\[
		0
		\le
		V_x^{n+1}(h)
		\le
		\bar u+\beta B
		=
		B .
	\]
	Thus \(0\le V_x^n(h)\le B\) for every \(n\) and \(h\).
	
	Define
	\[
		V_x(h):=\lim_{n\to\infty}V_x^n(h).
	\]
	The limit exists by monotone convergence. Since \(S\) is finite,
	\[
		\bigl(V_x^n(h,y)\bigr)_{y\in S}
		\to
		\bigl(V_x(h,y)\bigr)_{y\in S}.
	\]
	Continuity of \(\rho_{\bar x(h)}\) then permits passage to the limit in
	\eqref{eq:finite-envelope-proof}, yielding
	\[
		V_x(h)
		=
		\max\left\{
			0,\,
			u(\bar x(h))
			+
			\beta
			\rho_{\bar x(h)}
			\left(
				\bigl(V_x(h,y)\bigr)_{y\in S}
			\right)
		\right\}.
	\]
	
	Minimality follows from the same iteration. Let \(W:H(x)\to\R_+\) be any
	bounded nonnegative solution of \eqref{eq:history-snell}. Since
	\(V_x^0=0\le W\), monotonicity of the recursion implies \(V_x^n\le W\) for
	every \(n\). Taking limits gives \(V_x\le W\). Hence \(V_x\) is the minimal
	bounded nonnegative solution.
	
	It remains to identify the stopping supremum. For a continuation set
	\(C\subseteq H(x)\), let
	\[
		C^{(n)}:=C\cap\bigcup_{t=0}^{n-1}H_t(x)
	\]
	be the rule that follows \(C\) for the first \(n\) local dates and then forces
	stopping. Each \(C^{(n)}\) is prehistory closed and finite-depth. Therefore its
	value is bounded above by \(V_x^n((x))\), and hence by \(V_x((x))\). Conversely,
	for each \(n\), \(V_x^n((x))\) is the supremum over all finite-depth
	prehistory closed continuation sets with depth at most \(n\). Hence
	\[
		\sup_{\tau\le\infty}\mathcal U_x(\tau)
		=
		\sup_{n<\infty}V_x^n((x))
		=
		V_x((x)),
	\]
	where the completed value of an unbounded local stopping rule is understood as
	the limit of the values of its finite-depth truncations.
\end{proof}

\section{Proof of the Common-Tail Characterization}\label{app:cash}

\begin{proof}[Proof of Proposition~\ref{prop:cash}]
	Fix $x$ and suppose R7 holds. Let $\ell,m,\bar\ell,\bar m$ and $(r_y)_{y\in S}$ satisfy \eqref{eq:R7-premise}. Define
	\[
	a^y:=\Gamma_y(\bar m,r_y),
	\qquad
	b^y:=\Gamma_y(\bar\ell,r_y),
	\qquad
	v_y:=U_y(a^y),
	\qquad
	c:=u(\bar\ell)-u(\bar m).
	\]
	By the deterministic representation,
	\[
	d_x((\ell,\bar m))\sim_x d_x((m,\bar\ell))
	\quad\Longrightarrow\quad
	u(\ell)+\beta u(\bar m)=u(m)+\beta u(\bar\ell).
	\]
	Thus
	\[
	u(\ell)=u(m)+\beta c.
	\]
	
	Now apply the recursive representation at each next state $y$. Since $a^y$ and $b^y$ differ only by replacing the one-period consequence $\bar m$ with $\bar\ell$ before the same continuation profile $r_y$,
	\[
	U_y(b^y)
	=
	u(\bar\ell)+\beta\rho_y\bigl((U_z(r_y^z))_{z\in S}\bigr)
	=
	U_y(a^y)+c
	=
	v_y+c.
	\]
	Therefore the two schedules in \eqref{eq:R7-conclusion} have recursive values
	\[
	u(\ell)+\beta\rho_x(v)
	\qquad\text{and}\qquad
	u(m)+\beta\rho_x(v+c\1).
	\]
	R7 says that these values are equal. Substituting $u(\ell)=u(m)+\beta c$ gives
	\[
	u(m)+\beta c+\beta\rho_x(v)
	=
	u(m)+\beta\rho_x(v+c\1),
	\]
	and hence
	\[
	\rho_x(v+c\1)=\rho_x(v)+c.
	\]
	This is \eqref{eq:cash-attainable}.
	
	Conversely, suppose cash additivity holds for every attainable common scalar shift. Whenever \eqref{eq:R7-premise} holds, the same calculation gives $u(\ell)=u(m)+\beta c$, and the two composite schedules in \eqref{eq:R7-conclusion} have values
	\[
	u(m)+\beta c+\beta\rho_x(v)
	\qquad\text{and}\qquad
	u(m)+\beta\rho_x(v+c\1).
	\]
	Cash additivity makes these values equal. Hence the two composite schedules are indifferent, which is R7. If the attainable continuation domain is closed under all common translations $v\mapsto v+c\1$, the attainable-shift identity is exactly full cash additivity on that domain.
\end{proof}

\section{Proof of the Bandit Representation and Index Theorem}\label{app:bandit-index}

\begin{proof}[Proof of Theorem~\ref{thm:bandit-index}]
	The proof has three steps. The first constructs the product Bellman equation.
	The second constructs the charged one-schedule index and proves nestedness. The
	third verifies the max-index rule using the prevailing-charge form of the
	Weber--Whittle argument.

	First, consider the product-state Bellman equation. At a product state
	\[
	x=(x_1,\ldots,x_N)\in\mathbf S,
	\]
	a feasible calendar step chooses exactly one local clock to advance. If schedule
	\(i\) is advanced, the product state can only move to a state of the form
	\[
	(x_{-i},y_i)
	=
	(x_1,\ldots,x_{i-1},y_i,x_{i+1},\ldots,x_N),
	\qquad y_i\in S_i.
	\]
	All other coordinates are frozen. Therefore the continuation vector seen by
	schedule \(i\) is
	\[
	\bigl(V(x_{-i},y_i)\bigr)_{y_i\in S_i}.
	\]
	The one-schedule recursive representation gives current utility \(u(x_i)\) and
	evaluates this vector through \(\rho_i^{x_i}\). Maximizing over the unique local
	clock to be advanced gives exactly
	\[
	V(x)
	=
	\max_{1\le i\le N}
	\left\{
	u(x_i)
	+
	\beta
	\rho_i^{x_i}
	\left(
	\bigl(V(x_{-i},y_i)\bigr)_{y_i\in S_i}
	\right)
	\right\}.
	\]

	The same assumptions give well posedness. Since \(\rho_i^{x_i}\) is increasing
	and cash additive, it is nonexpansive in the sup norm. Indeed, if
	\(\|v-w\|_\infty\le a\), then
	\[
	v\le w+a\1
	\qquad\text{and}\qquad
	w\le v+a\1 .
	\]
	Monotonicity and cash additivity imply
	\[
	\rho_i^{x_i}(v)\le \rho_i^{x_i}(w)+a
	\qquad\text{and}\qquad
	\rho_i^{x_i}(w)\le \rho_i^{x_i}(v)+a.
	\]
	Thus
	\[
	|\rho_i^{x_i}(v)-\rho_i^{x_i}(w)|\le \|v-w\|_\infty .
	\]
	Because \(\beta\in(0,1)\), the product Bellman operator is a contraction on
	bounded functions on \(\mathbf S\). Hence \eqref{eq:bandit-bellman} has a
	unique bounded fixed point, which is the value of the elapsed-calendar
	constrained multi-schedule problem.

	Second, fix a schedule \(i\) and a scalar charge \(\lambda\). Define the excess
	value
	\[
	F_i^\lambda(x_i)
	:=
	V_i^\lambda(x_i)-\frac{\lambda}{1-\beta}.
	\]
	Using cash additivity in \eqref{eq:charged-stopping-value},
	\[
	\rho_i^{x_i}
	\left(
	\bigl(V_i^\lambda(y_i)\bigr)_{y_i\in S_i}
	\right)
	=
	\rho_i^{x_i}
	\left(
	\bigl(F_i^\lambda(y_i)\bigr)_{y_i\in S_i}
	+
	\frac{\lambda}{1-\beta}\1
	\right)
	=
	\rho_i^{x_i}
	\left(
	\bigl(F_i^\lambda(y_i)\bigr)_{y_i\in S_i}
	\right)
	+
	\frac{\lambda}{1-\beta}.
	\]
	Subtracting \(\lambda/(1-\beta)\) from both sides of
	\eqref{eq:charged-stopping-value} gives
	\begin{equation}\label{eq:appendix-excess}
	F_i^\lambda(x_i)
	=
	\max\left\{
	0,\,
	u(x_i)-\lambda
	+
	\beta
	\rho_i^{x_i}
	\left(
	\bigl(F_i^\lambda(y_i)\bigr)_{y_i\in S_i}
	\right)
	\right\}.
	\end{equation}
	Thus the charged on-hold problem is equivalently a charged stopping problem with
	stopping payoff zero and active payoff reduced by the scalar charge \(\lambda\).

	Let \(\mathcal T_i^\lambda\) denote the operator on the right-hand side of
	\eqref{eq:appendix-excess}. The nonexpansiveness argument above implies that
	\(\mathcal T_i^\lambda\) is a contraction. Hence \(F_i^\lambda\) is its unique
	fixed point. If \(\lambda'\ge\lambda\), then
	\[
	(\mathcal T_i^{\lambda'} f)(x_i)
	\le
	(\mathcal T_i^\lambda f)(x_i)
	\qquad
	\text{for every }f\text{ and every }x_i.
	\]
	By the fixed-point comparison theorem for monotone contractions,
	\[
	F_i^{\lambda'}(x_i)\le F_i^\lambda(x_i)
	\qquad
	\text{for every }x_i.
	\]
	Since each \(F_i^\lambda\) is nonnegative by \eqref{eq:appendix-excess}, the
	on-hold set can be written as
	\[
	D_i(\lambda)
	=
	\{x_i:F_i^\lambda(x_i)=0\}.
	\]
	Therefore,
	\[
	\lambda'\ge\lambda
	\quad\Longrightarrow\quad
	D_i(\lambda)\subseteq D_i(\lambda'),
	\]
	which is the nestedness assertion in the theorem.

	The endpoint and boundary properties follow from the same contraction argument.
	Because \(S_i\) is finite and the fixed point of a contraction depends
	continuously on parameters, \(F_i^\lambda(x_i)\) is continuous in \(\lambda\) for
	each \(x_i\). For \(\lambda\) sufficiently high, \(F\equiv0\) is the fixed point
	of \eqref{eq:appendix-excess}, so \(D_i(\lambda)=S_i\). For \(\lambda\)
	sufficiently low, one has \(F_i^\lambda(x_i)>0\) for every \(x_i\), so
	\(D_i(\lambda)=\varnothing\). Hence
	\[
	\Lambda_i(x_i)
	=
	\inf\{\lambda:x_i\in D_i(\lambda)\}
	\]
	is finite and well defined. Moreover, at \(\lambda=\Lambda_i(x_i)\), the active
	and on-hold branches in \eqref{eq:appendix-excess} are both optimal. Indeed,
	continuity gives \(F_i^{\Lambda_i(x_i)}(x_i)=0\); if the active branch were
	strictly below zero at this charge, it would remain below zero for nearby lower
	charges, contradicting the definition of \(\Lambda_i(x_i)\).

	Third, we verify the max-index rule. First record the equivalent passive-subsidy
	form of the same one-schedule problem. For each \(i\), write
	\[
	A_i^\lambda(x_i)
	:=
	u(x_i)-\lambda
	+
	\beta
	\rho_i^{x_i}
	\left(
	\bigl(F_i^\lambda(y_i)\bigr)_{y_i\in S_i}
	\right).
	\]
	Equation \eqref{eq:appendix-excess} says
	\[
	F_i^\lambda(x_i)=\max\{0,A_i^\lambda(x_i)\}.
	\]
	Since \(F_i^\lambda\ge0\), this is equivalent to
	\begin{equation}\label{eq:appendix-passive-excess}
	F_i^\lambda(x_i)
	=
	\max\left\{
	\beta F_i^\lambda(x_i),\,
	u(x_i)-\lambda
	+
	\beta
	\rho_i^{x_i}
	\left(
	\bigl(F_i^\lambda(y_i)\bigr)_{y_i\in S_i}
	\right)
	\right\}.
	\end{equation}
	Adding \(\lambda/(1-\beta)\) to \eqref{eq:appendix-passive-excess} gives
	\begin{equation}\label{eq:appendix-passive-subsidy}
	V_i^\lambda(x_i)
	=
	\max\left\{
	\lambda+\beta V_i^\lambda(x_i),\,
	u(x_i)
	+
	\beta
	\rho_i^{x_i}
	\left(
	\bigl(V_i^\lambda(y_i)\bigr)_{y_i\in S_i}
	\right)
	\right\}.
	\end{equation}
	Thus \(V_i^\lambda\) is also the one-schedule value when passivity pays the
	scalar subsidy \(\lambda\) and leaves the rested state fixed. The index
	\(\Lambda_i(x_i)\) is therefore the critical subsidy at which schedule \(i\), in
	state \(x_i\), first becomes willing to remain passive.

	For any fixed charge \(\lambda\), define the separated charged upper bound
	\begin{equation}\label{eq:appendix-charged-bound}
	B^\lambda(x)
	:=
	\frac{\lambda}{1-\beta}
	+
	\sum_{i=1}^N F_i^\lambda(x_i).
	\end{equation}
	This is the value obtained after pricing one unit of elapsed local time at
	\(\lambda\) and then allowing the one-schedule charged stopping problems to
	separate. It is an upper bound on the original elapsed-calendar value. To see
	this directly, fix any product state \(z\in\mathbf S\) and any schedule \(k\).
	Since
	\[
	B^\lambda(z_{-k},y_k)
	=
	\frac{\lambda}{1-\beta}
	+
	\sum_{\ell\ne k}F_\ell^\lambda(z_\ell)
	+
	F_k^\lambda(y_k),
	\]
	cash additivity and \eqref{eq:appendix-excess} imply
	\begin{align*}
	&u(z_k)
	+
	\beta
	\rho_k^{z_k}
	\left(
	\bigl(B^\lambda(z_{-k},y_k)\bigr)_{y_k\in S_k}
	\right) \\
	&\qquad =
	u(z_k)
	+
	\beta
	\rho_k^{z_k}
	\left(
	\bigl(F_k^\lambda(y_k)\bigr)_{y_k\in S_k}
	\right)
	+
	\beta
	\left[
	\frac{\lambda}{1-\beta}
	+
	\sum_{\ell\ne k}F_\ell^\lambda(z_\ell)
	\right] \\
	&\qquad \le
	F_k^\lambda(z_k)+\lambda
	+
	\beta
	\left[
	\frac{\lambda}{1-\beta}
	+
	\sum_{\ell\ne k}F_\ell^\lambda(z_\ell)
	\right] \\
	&\qquad =
	\frac{\lambda}{1-\beta}
	+
	F_k^\lambda(z_k)
	+
	\beta\sum_{\ell\ne k}F_\ell^\lambda(z_\ell) \\
	&\qquad \le
	\frac{\lambda}{1-\beta}
	+
	\sum_{\ell=1}^N F_\ell^\lambda(z_\ell)
	=
	B^\lambda(z).
	\end{align*}
	Thus \(B^\lambda\) is a Bellman supersolution for the original product problem,
	and therefore upper-bounds the value of every original feasible policy.

	Now fix a product state \(x\). Suppose first that the largest current index is
	unique. Let \(j\) be the unique schedule with
	\[
	\Lambda_j(x_j)>\max_{i\ne j}\Lambda_i(x_i),
	\]
	and choose \(\lambda\) strictly between the largest and second-largest current
	indices. Then
	\[
	F_i^\lambda(x_i)=0\quad\text{for every }i\ne j,
	\qquad
	F_j^\lambda(x_j)>0.
	\]
	For schedule \(j\), the active branch binds in \eqref{eq:appendix-excess}:
	\[
	F_j^\lambda(x_j)
	=
	u(x_j)-\lambda
	+
	\beta\rho_j^{x_j}
	\left(
	\bigl(F_j^\lambda(y_j)\bigr)_{y_j\in S_j}
	\right).
	\]
	Using cash additivity and the fact that all non-\(j\) excess values at the
	current state are zero,
	\begin{align*}
	&u(x_j)
	+
	\beta
	\rho_j^{x_j}
	\left(
	\bigl(B^\lambda(x_{-j},y_j)\bigr)_{y_j\in S_j}
	\right) \\
	&\qquad =
	u(x_j)
	+
	\beta
	\rho_j^{x_j}
	\left(
	\bigl(F_j^\lambda(y_j)\bigr)_{y_j\in S_j}
	+
	\frac{\lambda}{1-\beta}\1
	\right) \\
	&\qquad =
	u(x_j)
	+
	\beta
	\rho_j^{x_j}
	\left(
	\bigl(F_j^\lambda(y_j)\bigr)_{y_j\in S_j}
	\right)
	+
	\frac{\beta\lambda}{1-\beta} \\
	&\qquad =
	F_j^\lambda(x_j)+\lambda+
	\frac{\beta\lambda}{1-\beta}
	=
	F_j^\lambda(x_j)+\frac{\lambda}{1-\beta}
	=
	B^\lambda(x).
	\end{align*}
	Thus advancing the unique largest-index schedule attains the charged upper bound
	in the current Bellman comparison.

	For any \(k\ne j\), the on-hold branch binds weakly in
	\eqref{eq:appendix-excess}:
	\[
	0
	\ge
	u(x_k)-\lambda
	+
	\beta\rho_k^{x_k}
	\left(
	\bigl(F_k^\lambda(y_k)\bigr)_{y_k\in S_k}
	\right).
	\]
	Again using cash additivity,
	\begin{align*}
	&u(x_k)
	+
	\beta
	\rho_k^{x_k}
	\left(
	\bigl(B^\lambda(x_{-k},y_k)\bigr)_{y_k\in S_k}
	\right) \\
	&\qquad =
	u(x_k)
	+
	\beta
	\rho_k^{x_k}
	\left(
	\bigl(F_k^\lambda(y_k)\bigr)_{y_k\in S_k}
	+
	\left[
	\frac{\lambda}{1-\beta}
	+
	F_j^\lambda(x_j)
	\right]\1
	\right) \\
	&\qquad \le
	\lambda
	+
	\beta
	\left[
	\frac{\lambda}{1-\beta}
	+
	F_j^\lambda(x_j)
	\right] \\
	&\qquad =
	\frac{\lambda}{1-\beta}
	+
	\beta F_j^\lambda(x_j)
	\le
	\frac{\lambda}{1-\beta}
	+
	F_j^\lambda(x_j)
	=
	B^\lambda(x).
	\end{align*}
	Thus no lower-index schedule can exceed the branch value obtained by advancing
	\(j\).

	If the largest current index is tied, let
	\[
	\bar\lambda(x):=\max_{1\le i\le N}\Lambda_i(x_i)
	\]
	and choose any \(j\in\argmax_i\Lambda_i(x_i)\). At
	\(\lambda=\bar\lambda(x)\), every maximal schedule is at its boundary, so the
	active and on-hold branches are both optimal for those schedules; every
	nonmaximal schedule has the on-hold branch weakly optimal. The same calculation
	therefore gives equality for the selected maximal schedule \(j\) and weak
	inequality for all other schedules. Hence any fixed tie-breaking rule among
	maximal-index schedules is valid.

	This is the prevailing-charge form of the Weber--Whittle argument. The family of
	one-schedule subsidy problems is solved for all \(\lambda\). At a product state,
	the relevant charge is chosen from the current vector of critical charges: if the
	maximal index is unique, choose any charge between the largest and second-largest
	indices; if it is tied, use the common boundary charge and the fixed
	tie-breaking rule. After the selected schedule moves and the product state
	changes, the prevailing charge is recomputed from the new current indices.

	The selector constructed in this way advances exactly one schedule at every
	product state and leaves all other rested schedules fixed. It is therefore
	feasible for the original elapsed-calendar constrained problem. The separated
	charged problem is a relaxation, hence an upper bound, of the original problem;
	the prevailing-charge calculation above selects a relaxed optimizer that satisfies
	the original elapsed-calendar constraint. The relaxation is therefore tight, and
	the resulting original-feasible policy is optimal. This policy is precisely the
	max-index rule in \eqref{eq:max-index-policy}.
\end{proof}
\section{Proofs for Extensions and Special Cases}\label{app:special-cases}

Throughout this appendix, D+R gives the recursive stopping representation from
Theorem~\ref{thm:recursive}. Thus continuation families can be identified, on the
rich continuation domain, with their continuation-value vectors
\[
v=(U_y(a^y))_{y\in S}\in\R^S .
\]
The deterministic normalization fixes constants:
\begin{equation}\label{eq:app-constant-normalization}
	\rho_x(t\1)=t .
\end{equation}
Whenever R7 is imposed, Proposition~\ref{prop:cash} gives cash additivity:
\begin{equation}\label{eq:app-cash}
	\rho_x(v+c\1)=\rho_x(v)+c .
\end{equation}
Consequently, if a one-time retirement value is \(W/(1-\beta)\), the charged
one-arm stopping problem can be written in excess-value form as
\begin{equation}\label{eq:app-excess-template}
	F^W(x)
	=
	\max\left\{
	0,\,
	u(x)-W+\beta\rho_x(F^W)
	\right\}.
\end{equation}
The associated index is
\[
I(x)=\inf\{W:F^W(x)=0\}.
\]

\begin{proof}[Proof of Proposition~\ref{prop:eu-index}]
	Fix \(x\). Under D+R, the induced continuation preference \(\succeq_x^C\)
	is represented by \(\rho_x\) on continuation-value vectors. EU is the
	finite-state Anscombe--Aumann independence axiom applied to these continuation
	acts. Together with weak order, continuity, and monotonicity from the recursive
	one-step axioms, it implies that \(\succeq_x^C\) has an affine expected-utility
	representation:
	\[
	v\mapsto a_x\sum_{y\in S}\pi_x(y)v_y+b_x,
	\]
	where \(\pi_x(y)\ge0\) and the weights are not all zero. Normalizing the weights
	gives a probability \(p_x\in\Delta(S)\).
	
	Since \(\rho_x\) represents the same continuation preference, there is a
	strictly increasing transform \(\varphi_x\) such that
	\[
	\rho_x(v)
	=
	\varphi_x\left(\sum_{y\in S}p_x(y)v_y\right).
	\]
	The constant normalization \eqref{eq:app-constant-normalization} gives
	\[
	t
	=
	\rho_x(t\1)
	=
	\varphi_x(t),
	\]
	so \(\varphi_x\) is the identity on the relevant range. Hence
	\[
	\rho_x(v)=\sum_{y\in S}p_x(y)v_y.
	\]
	This proves the expected-utility representation. Conversely, a linear
	expectation over a Markov kernel satisfies continuation independence, and it is
	cash additive. Thus the expected-utility case is the sharp iff specialization
	of the recursive representation.
	
	Substituting the linear aggregator into \eqref{eq:app-excess-template} gives
	\eqref{eq:eu-excess-special}. For the ratio formula, fix
	\(\tau\in\mathfrak T_x^+\). Continuing until \(\tau\) and then retiring gives
	excess value
	\[
	\E_x^p\left[
	\sum_{t=0}^{\tau-1}\beta^t(u(X_t)-W)
	\right]
	=
	\E_x^p\left[
	\sum_{t=0}^{\tau-1}\beta^t u(X_t)
	\right]
	-
	W
	\E_x^p\left[
	\sum_{t=0}^{\tau-1}\beta^t
	\right].
	\]
	This continuation plan is weakly worthwhile exactly when \(W\) is no larger
	than the corresponding discounted reward-per-time ratio. Taking the supremum
	over positive finite stopping times gives \eqref{eq:eu-index-special}.
\end{proof}

\begin{proof}[Proof of Lemma~\ref{lem:finite-good-news-representation}]
	By D, R, and EU, the continuation aggregator at each active integer state
	\(n=1,\ldots,T\) is linear on the finite continuation domain. Hence there is a
	probability vector
	\[
		q_n\in\Delta(S_T^G)
	\]
	such that, for every continuation-value vector \(v\in\mathbb R^{S_T^G}\),
	\[
		\rho_n(v)=\sum_{y\in S_T^G}q_n(y)v(y).
	\]
	For \(A\subseteq S_T^G\), write
	\[
		q_n(A):=\sum_{y\in A}q_n(y).
	\]

	First, the state-ordering part of GNL identifies the unresolved states as
	strictly between the two benchmark states. Indeed, for any active integer state
	\(n\), the constant continuation family \([z]_n\) delivers the same consequence
	\(z\) at every next-state coordinate. Constant preservation of the continuation
	aggregator therefore gives the value \(u(z)\). Thus
	\[
		[\rho]_n\succ_n^C [j]_n\succ_n^C[\gamma]_n
	\]
	implies
	\[
		u(\rho)>u(j)>u(\gamma)
		\qquad
		\text{for every }j=0,\ldots,T.
	\]
	In particular \(u(\rho)>u(\gamma)\), so the normalized values
	\[
		\mu_j:=\frac{u(j)-u(\gamma)}{u(\rho)-u(\gamma)}
	\]
	are well defined and satisfy
	\[
		\mu_j\in(0,1),
		\qquad j=0,\ldots,T.
	\]

	Next, consider the exhaustiveness condition. The continuation family
	\[
		[\rho\text{ on }\{\rho,n-1\};\gamma\text{ off }\{\rho,n-1\}]_n
	\]
	has represented continuation value
	\[
		q_n(\{\rho,n-1\})u(\rho)
		+
		\bigl(1-q_n(\{\rho,n-1\})\bigr)u(\gamma).
	\]
	The constant continuation family \([\rho]_n\) has value \(u(\rho)\). The
	indifference in \eqref{eq:gnl-exhaustiveness} therefore gives
	\[
		q_n(\{\rho,n-1\})u(\rho)
		+
		\bigl(1-q_n(\{\rho,n-1\})\bigr)u(\gamma)
		=
		u(\rho).
	\]
	Since \(u(\rho)>u(\gamma)\), this implies
	\[
		q_n(\{\rho,n-1\})=1.
	\]
	Hence, under the represented one-step law at \(n\), all weight is placed on the
	two coordinates \(\rho\) and \(n-1\).

	Now consider the stationary relative-price condition. The good-news branch bet
	\[
		[\rho\text{ on }\{\rho\};\gamma\text{ off }\{\rho\}]_n
	\]
	has represented value
	\[
		q_n(\rho)u(\rho)+\bigl(1-q_n(\rho)\bigr)u(\gamma).
	\]
	The mixed constant continuation family
	\[
		\alpha[n]_n+(1-\alpha)[\gamma]_n
	\]
	has represented value
	\[
		\alpha u(n)+(1-\alpha)u(\gamma).
	\]
	Therefore
	\[
		[\rho\text{ on }\{\rho\};\gamma\text{ off }\{\rho\}]_n
		\succeq_n^C
		\alpha[n]_n+(1-\alpha)[\gamma]_n
	\]
	is equivalent to
	\[
		q_n(\rho)\bigl(u(\rho)-u(\gamma)\bigr)
		\ge
		\alpha\bigl(u(n)-u(\gamma)\bigr),
	\]
	or, equivalently,
	\[
		q_n(\rho)\ge \alpha\mu_n.
	\]
	The strict comparison
	\[
		[n]_n
		\succ_n^C
		[\rho\text{ on }\{\rho\};\gamma\text{ off }\{\rho\}]_n
		\succ_n^C
		[\gamma]_n
	\]
	implies
	\[
		0<q_n(\rho)<\mu_n.
	\]
	Hence \(q_n(\rho)/\mu_n\in(0,1)\).

	The equivalence in \eqref{eq:gnl-relative-stationarity}, applied for all
	\(\alpha\in(0,1)\), implies that the threshold
	\[
		\frac{q_n(\rho)}{\mu_n}
	\]
	is the same at every active integer state. Otherwise, if
	\(q_n(\rho)/\mu_n\ne q_m(\rho)/\mu_m\), one could choose
	\(\alpha\in(0,1)\) strictly between the two ratios and violate the stated
	equivalence. Let the common value be
	\[
		\theta\in(0,1).
	\]
	Then, for every \(n=1,\ldots,T\),
	\[
		q_n(\rho)=\theta\mu_n.
	\]
	Using \(q_n(\{\rho,n-1\})=1\), we also obtain
	\[
		q_n(n-1)=1-\theta\mu_n.
	\]

	It remains to derive the no-news update. By successor-consequence consistency,
	\[
		[n]_n
		\sim_n^C
		[\rho\text{ on }\{\rho\};\, n-1\text{ on }\{n-1\};\, \gamma
		\text{ off }\{\rho,n-1\}]_n .
	\]
	The left-hand side has value \(u(n)\). Since the represented law places all
	weight on \(\{\rho,n-1\}\), the right-hand side has value
	\[
		q_n(\rho)u(\rho)+q_n(n-1)u(n-1).
	\]
	Thus
	\[
		u(n)
		=
		q_n(\rho)u(\rho)+q_n(n-1)u(n-1).
	\]
	Subtract \(u(\gamma)\) from both sides and divide by
	\(u(\rho)-u(\gamma)\). Using the definitions of \(\mu_n\) and
	\(\mu_{n-1}\), this gives
	\[
		\mu_n
		=
		q_n(\rho)+(1-q_n(\rho))\mu_{n-1}.
	\]
	Substituting \(q_n(\rho)=\theta\mu_n\), we get
	\[
		\mu_n
		=
		\theta\mu_n+(1-\theta\mu_n)\mu_{n-1}.
	\]
	Solving for \(\mu_{n-1}\) yields
	\[
		\mu_{n-1}
		=
		\frac{(1-\theta)\mu_n}{1-\theta\mu_n},
	\]
	which is \eqref{eq:gnl-no-news-update}.

	Equivalently,
	\[
		\frac{\mu_{n-1}}{1-\mu_{n-1}}
		=
		(1-\theta)\frac{\mu_n}{1-\mu_n}.
	\]
	Taking logarithms gives
	\[
		\log\frac{\mu_{n-1}}{1-\mu_{n-1}}
		=
		\log\frac{\mu_n}{1-\mu_n}
		+
		\log(1-\theta),
	\]
	which is \eqref{eq:gnl-log-odds-update}. Hence, if
	\[
		r_n:=\log\frac{\mu_n}{1-\mu_n},
	\]
	then
	\[
		r_{n-1}=r_n+\log(1-\theta).
	\]
	Writing
	\[
		\delta:=-\log(1-\theta)>0,
	\]
	we have
	\[
		r_n=a+n\delta
	\]
	for the constant \(a=r_0\). Thus the primitive integer labels are represented as
	an affine log-likelihood scale, and the represented one-step law has the
	stationary good-news form
	\[
		q_n(\rho)=\theta\mu_n,
		\qquad
		q_n(n-1)=1-\theta\mu_n.
	\]
	This is exactly the finite good-news process in which good news arrives with
	probability \(\theta\) conditional on the good type and never arrives
	conditional on the bad type.
\end{proof}

\begin{proof}[Proof of Proposition~\ref{prop:variational-index}]
	Fix \(x\). By R7 and Proposition~\ref{prop:cash}, \(\rho_x\) is cash additive.
	VA says that if two continuation-value vectors are indifferent, then their
	objective mixture is weakly preferred. Thus \(\rho_x\) is quasiconcave. Cash
	additivity turns quasiconcavity into concavity. Let
	\[
	a=\rho_x(v),
	\qquad
	b=\rho_x(w).
	\]
	Then
	\[
	\rho_x(v-a\1)=0=\rho_x(w-b\1).
	\]
	By VA,
	\[
	\rho_x\bigl(\alpha(v-a\1)+(1-\alpha)(w-b\1)\bigr)\ge0.
	\]
	Using cash additivity,
	\[
	\rho_x(\alpha v+(1-\alpha)w)
	\ge
	\alpha\rho_x(v)+(1-\alpha)\rho_x(w).
	\]
	Thus \(\rho_x\) is concave, monotone, continuous, cash additive, and normalized.
	
	Let \(\phi_x(v):=-\rho_x(v)\). Then \(\phi_x\) is convex, continuous,
	decreasing, and satisfies
	\[
	\phi_x(v+c\1)=\phi_x(v)-c.
	\]
	By finite-dimensional convex duality,
	\[
	\phi_x(v)
	=
	\sup_{q\in\R^S}\{q\cdot v-\phi_x^*(q)\}.
	\]
	If \(\phi_x^*(q)<\infty\), cash additivity forces \(q\cdot\1=-1\), while
	monotonicity forces \(q\le0\). Hence \(q=-p\) for some \(p\in\Delta(S)\).
	Define
	\[
	\alpha_x(p):=\phi_x^*(-p).
	\]
	Then
	\[
	\rho_x(v)
	=
	\min_{p\in\Delta(S)}\{p\cdot v+\alpha_x(p)\}.
	\]
	The penalty is convex and lower semicontinuous by construction, and
	\(\rho_x(0)=0\) implies \(\min_p\alpha_x(p)=0\).
	
	Conversely, any aggregator of the displayed variational form is monotone,
	cash additive, and concave:
	\[
	\min_p\{p\cdot(v+c\1)+\alpha_x(p)\}
	=
	c+\min_p\{p\cdot v+\alpha_x(p)\}.
	\]
	Cash additivity gives R7 through Proposition~\ref{prop:cash}, and concavity
	gives VA. Thus the variational representation is the iff form of D+R+R7+VA
	on the rich continuation domain.
	
	Substituting the variational aggregator into \eqref{eq:app-excess-template}
	gives \eqref{eq:variational-excess-special}. For the ratio expression, fix
	\(\tau\in\mathfrak T_x^+\). Under the rectangular interpretation of one-step
	penalties, continuing until \(\tau\) gives excess value
	\[
	\inf_{\mathbb P}
	\E_{\mathbb P}\left[
	\sum_{t=0}^{\tau-1}\beta^t(u(X_t)-W)
	+
	\sum_{t=0}^{\tau-1}\beta^{t+1}\alpha_{X_t}(p_t)
	\right].
	\]
	The \(W\)-terms separate as
	\[
	-W
	\E_{\mathbb P}\left[
	\sum_{t=0}^{\tau-1}\beta^t
	\right].
	\]
	Thus the plan is weakly worthwhile exactly when \(W\) is no larger than the
	penalized ratio in \eqref{eq:variational-index-special}. Taking the supremum
	over stopping times gives the variational index.
\end{proof}

\begin{proof}[Proof of Proposition~\ref{prop:choquet-index}]
	Fix \(x\). Under D+R, continuation schedules are represented on
	continuation-value vectors by \(\rho_x\). CEU imposes independence only on
	comonotonic continuation families. By the finite-state Schmeidler representation
	theorem, continuity, monotonicity, and comonotonic independence imply that
	\(\succeq_x^C\) is represented by a Choquet integral with respect to a
	normalized capacity \(\nu_x\):
	\[
	\rho_x(v)=\int v\,d\nu_x.
	\]
	The deterministic normalization pins down constants, so the integral is on the
	same utility scale fixed by D+R.
	
	Conversely, a finite-state Choquet integral with respect to a normalized
	capacity is continuous, monotone, and satisfies comonotonic independence. Thus
	the Choquet representation is the iff form of the CEU restriction on the
	continuation domain.
	
	A normalized Choquet integral is cash additive:
	\[
	\int (v+c\1)\,d\nu_x
	=
	\int v\,d\nu_x+c.
	\]
	Substituting the Choquet aggregator into \eqref{eq:app-excess-template} gives
	\eqref{eq:choquet-excess-special}, and the index is
	\[
	I^{CH}(x)=\inf\{W:F^W(x)=0\}.
	\]
	If \(\nu_x\) is additive, then
	\[
	\int v\,d\nu_x=\sum_{y\in S}\nu_x(\{y\})v_y,
	\]
	so the model reduces to the expected-utility case.
\end{proof}

\begin{proof}[Proof of Proposition~\ref{prop:pandora-index}]
	Fix an unopened box \(i\). If the outside continuation value is \(b\), opening
	the box gives
	\[
	H_i(b):=-c_i+\beta\rho_i\bigl((b\vee g_i(z))_{z\in Z_i}\bigr).
	\]
	A reservation value is a fixed point \(b=H_i(b)\). Since \(\rho_i\) is monotone
	and cash additive, it is nonexpansive in the sup norm: if
	\(\|v-w\|_\infty\le a\), then
	\[
	v\le w+a\1
	\quad\text{and}\quad
	w\le v+a\1,
	\]
	so monotonicity and cash additivity imply
	\[
	|\rho_i(v)-\rho_i(w)|\le a.
	\]
	The map \(b\mapsto (b\vee g_i(z))_{z\in Z_i}\) is one-Lipschitz, so \(H_i\) is
	\(\beta\)-Lipschitz. Since \(\beta<1\), \(H_i\) has a unique fixed point
	\(\pi_i\), which satisfies \eqref{eq:pandora-reservation}.
	
	Cash additivity gives the equivalent excess form. Since
	\[
	\pi_i\vee g_i(z)
	=
	\pi_i+(g_i(z)-\pi_i)^+,
	\]
	we have
	\[
	\rho_i\bigl((\pi_i\vee g_i(z))_{z\in Z_i}\bigr)
	=
	\pi_i+\rho_i\bigl(((g_i(z)-\pi_i)^+)_{z\in Z_i}\bigr).
	\]
	Substituting this into \eqref{eq:pandora-reservation} gives
	\[
	\pi_i
	=
	-c_i+\beta\pi_i
	+
	\beta\rho_i\bigl(((g_i(z)-\pi_i)^+)_{z\in Z_i}\bigr),
	\]
	which is equivalent to \eqref{eq:pandora-reservation-cash}.
	
	For a revealed state \(z\), the local state is absorbing, so using the revealed
	box forever gives value
	\[
	g_i(z)=\frac{u(z)}{1-\beta}.
	\]
	The corresponding per-period index is \(u(z)=(1-\beta)g_i(z)\). At the
	unopened state, the one-time reservation value is \(\pi_i\), so the per-period
	index is \((1-\beta)\pi_i\). This proves \eqref{eq:pandora-index-special}.
	
	It remains to verify the Pandora rule. Let \(G\) be the best currently revealed
	absorbing value. Opening box \(i\) is worthwhile against outside value \(G\)
	exactly when
	\[
	-c_i+\beta\rho_i\bigl((G\vee g_i(z))_{z\in Z_i}\bigr)\ge G.
	\]
	The left-hand side minus \(G\) is weakly decreasing in \(G\) and is equal to
	zero at \(G=\pi_i\). Hence box \(i\) is worth opening exactly when
	\(G\le\pi_i\), with strict preference when \(G<\pi_i\). Therefore, if the best
	revealed value is below the largest reservation value, the optimal action is to
	open a box with maximal \(\pi_i\). If the best revealed value is at least as
	large as every reservation value, the optimal action is to commit to the best
	revealed absorbing state. Backward induction on the number of unopened boxes
	gives the stated rule.
	
	Under subjective expected utility,
	\[
	\rho_i(v)=\sum_{z\in Z_i}p_i(z)v(z),
	\]
	so \eqref{eq:pandora-reservation-cash} becomes
	\eqref{eq:pandora-seu-reservation-special}. With \(g_i(z)=z\) and the formal
	boundary case \(\beta=1\), this reduces to Weitzman's reservation equation.
\end{proof}

\

\bibliographystyle{plainnat}
\bibliography{bandit_literature}

\end{document}